\documentclass[12pt]{article}
\usepackage{jheppub_modified_green}
\usepackage{graphicx}
\usepackage{epsfig}
\usepackage{amsfonts}
\usepackage{amssymb}
\usepackage{amsmath}
\usepackage{amsthm}
\usepackage{subfig}
\usepackage{bm}
\usepackage{hyperref}
\usepackage{mathrsfs}
\usepackage{bbm}
\usepackage{cancel}
\usepackage{color}
\usepackage{accents}

\def\be{\begin{equation}}
\def\ee{\end{equation}}
\def\bea{\begin{eqnarray}}
\def\eea{\end{eqnarray}}
\newcommand{\beas}{\begin{eqnarray*}}
\newcommand{\eeas}{\end{eqnarray*}}

\newcommand{\nb}{\nonumber}

\newtheorem{theorem}{\sf THEOREM}

\def\Label#1{\label{#1}%
  \smash{\hbox to0pt{\raise1ex\hbox{\tiny[#1]}\hss}}}

\allowdisplaybreaks
\title{Syzygies Probing Scattering Amplitudes}

\author[a,b]{Gang Chen}
\author[c,d]{Junyu Liu}
\author[b]{Ruofei Xie}
\author[b]{Hao Zhang}
\author[c,d]{Yehao Zhou}
\affiliation[a]{Department of Physics, Zhejiang Normal University 688 Yingbin Road, Jinhua 321004, China}
\affiliation[b]{Department of Physics, Nanjing University 22 Hankou Road, Nanjing 210093, China}
\affiliation[c]{School of the Gifted Young, University of Science and Technology of China, Hefei, Anhui 230026, China}
\affiliation[d]{School of Physical Sciences, University of Science and Technology of China, Hefei, Anhui 230026, China}
\emailAdd{gang.chern@gmail.com}
\emailAdd{junyu@mail.ustc.edu.cn}
\emailAdd{rfxie@yangian.com}
\emailAdd{hao.zhang.phy@gmail.com}
\emailAdd{zyh12203@mail.ustc.edu.cn}

\date{\today}
\abstract{We propose a new efficient algorithm to obtain the locally minimal generating set of the syzygies for an ideal, i.e. a generating set whose proper subsets cannot be generating sets. Syzygy is a concept widely used in the current study of scattering amplitudes. This new algorithm  can deal with more syzygies effectively because a new generation of syzygies is obtained in each step and the irreducibility of this generation is also verified in the process. This efficient algorithm can also be applied in getting the syzygies for the modules. We also show a typical example to illustrate the potential application of this method in scattering amplitudes, especially the Integral-By-Part(IBP) relations of the characteristic two-loop diagrams in the Yang-Mills theory.}

\begin{document}
\maketitle
\section{Introduction}
Scattering amplitude is a leading research area which has various applications in phenomenology and attracts attention in a wide range of formal theories \cite{Britto:2005fq,ArkaniHamed:2012nw}. The major objects we focus on in the field of scattering amplitude are polynomial functions and rational functions. In modern mathematics, an efficient tool of dealing with rational functions stems from the theoretical structure of algebraic geometry \cite{math3}. Thus, the study of scattering amplitudes may bring us a fascinating connection between mathematics and physics. Many concepts and methods in algebraic geometry play a significant role in both the calculation and the theoretical analysis in scattering amplitude. One of the most widely used concepts is \emph{syzygy} \cite{math2,math5,math10,math11}, which is the relation set of an m-tuple polynomial function.

Up to now, syzygies have appeared in lots of research topics in scattering amplitude, such as the IBP relation \cite{Gluza:2010ws,Schabinger:2011dz,Zhang:2014xwa,Grozin:2011mt,Kosower:2011ty,Ita:2015tya,Larsen:2015ped}, which is used to determine irreducible loop integrals. It can also be used to fix the ambiguity of the integrands of the non-planar amplitudes. Another application is to simplify the Grassmannian integral form of the $\mathcal{N}=4$ Super Yang-Mills non-planar amplitude \cite{Chen:2014ara,Arkani-Hamed:2014bca,Franco:2015rma,Chen:2015bnt,Benincasa:2015zna,Frassek:2015rka,Bern:2014kca}. Furthermore, syzygies can potentially be used to construct the loop-level scattering amplitudes from unitarity cuts \cite{Du:2014jwa,Bern:1994zx} and to probe the amplitude relations beyond the KK-relation \cite{Kleiss:1988ne} and the BCJ-relation \cite{Bern:2008qj} in Yang-Mills theory.

We usually need a highly efficient algorithm to obtain syzygies in most applications in physics. To give compact expressions for the quantities of interest, such as the IBP relations of the loop integrals, we also need to get the locally minimal generating set for the syzygies. In this paper,we develop an effective algorithm to obtain the irreducible basis of the syzygies of an ideal. Meanwhile, the Gr\"{o}bner basis of the ideal is also obtained. Current algorithms, such as MMT \cite{math13} and F5 \cite{math6}, are very efficient to get the Gr\"{o}bner basis \cite{Gluza:2010ws,Schabinger:2011dz,math4,math6,math7,math8,math9,math10,math11} of syzygies or ideals. However, an algorithm to reduce the number of necessary syzygies is still needed. In \cite{Gluza:2010ws}, Gluza et al also considered reducing a new syzygy obtained by the known syzygies, recursively, in their algorithm. The major advantage of our method is that we deal with more syzygies effectively and efficient feed-back is used to verify the irreducibility of each new syzygy added. This leads to that all the irreducible syzygies can be obtained even before the end of the main loop. This algorithm, with a promising significance in both mathematics and theoretical physics, can also be used in other areas that need to classify irreducible algebraic relations.  The major difference between this algorithm and MMT \cite{math13}  is that we only use the leading terms of polynomials to justify the permitted critical pairs and reduce the generated syzygies in each step. At each step, we can feedback to up-levels to guarantee that our justification is correct. This difference also makes our algorithm faster than MMT in practice.

This paper is organized as following. In Section \ref{sec1} we will first give a warming-up example, and then illustrate our new method for syzygies of ideals in detail with the generalization to module cases. Section \ref{sec2} will focus on the application of this algorithm in the IBP relations of a specific two-loop diagram. In Section \ref{sec3}, we will give some conclusions and remarks on this new algorithm.

\section{A Direct Method To Compute Syzygies}\label{sec1}
A syzygy of a polynomial m-tuple  $\mathbf{f}=(f_1, f_2\cdots f_m)$ is a m-tuple $\mathbf{S}=(a_1,\cdots,a_m)$  such that $\mathbf{S}\cdot\mathbf{f}=0$. The traditional method to get the syzygies of an ideal is based on the Gr\"{o}ebner basis techniques. The locally minimal generating set of the obtained syzygies can only be obtained by the syzygies of the syzygy module. This algorithm is usually not so efficient. Another method is based on the linear algebra techniques \cite{math91} which does not rely on the Gr\"{o}ebner basis. This method does not guarantee the completeness of the syzygies. In this section, we propose a new method to formulate all the syzygies completely. Meanwhile our method guarantees that the generating set  of the syzygy module is locally minimal. The Gr\"{o}ebner basis is obtained automatically in the process.

The general strategy to get the locally minimal generating set of syzygies is to decompose these syzygies by the syzygy of the leading monomials.  This is realized by computing the S-polynomial from a chosen critical pair  step by step. In each step, this generates an S-polynomial $S_{i,j}={lcm(i,j)\over LT(f_i)} f_i-{lcm(i,j)\over LT(f_j)} f_j$ with new leading terms. Finally if there exists a vanishing $S_{i,j}$, we can get the syzygy of the original m-tuple by combining all the two-pairs inversely as shown in Fig. \ref{Fig-TwoPairDec}.
\begin{figure}
  \centering
  \includegraphics[width=12cm]{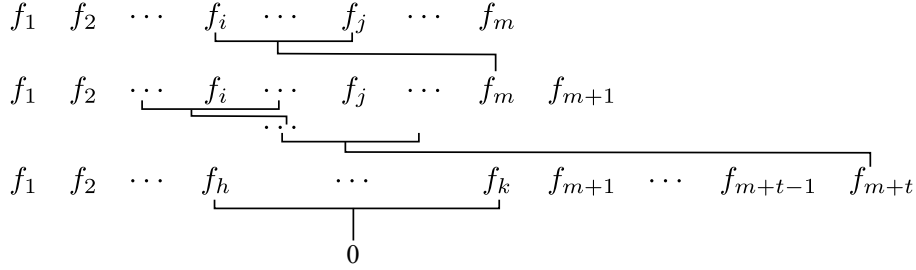}\\
  \caption{Critical pair decomposition.}  \label{Fig-TwoPairDec}
\end{figure}
We give some definitions first:
\begin{itemize}
\item $f$: A polynomial in the polynomial ring $R=K[x_1\cdots x_n]$, where $K$ is an algebraic closed field.
\item $>$: A monomial order on the polynomial ring $R$. In this paper we usually choose the Degree Reverse Lexicographic order. Let $x^\alpha$ and $x^\beta$ be monomials in $R$, where $\alpha(\beta)$ is the exponent vector $\alpha(\beta)\in \mathbb{Z}^n$. We say $x^\alpha>x^\beta$ if $\sum_{i=1}^n\alpha_i>\sum_{i=1}^n\beta_i$ or if $\sum_{i=1}^n\alpha_i=\sum_{i=1}^n\beta_i$, and in the difference $\alpha-\beta$, the rightmost nonzero entry is negative.
\item$\succ$: For a polynomial m-tuple, we define an additional order $\vec{e}_1\succ\vec{e}_2\cdots\succ\vec{e}_m$. The monomial order for this m-tuple is a POT extension of $>$: We say $x^\alpha e_i\succ x^\beta e_j$ if $i<j$ or if $i=j$ and $x^\alpha>x^\beta$
\item $I=\langle f_1\cdots f_m\rangle$: An ideal generated by $f_1 \cdots f_m$ in $R$, where $m\in \mathbb{Z}_{>0}$.
\item $\langle g_1\cdots g_m\rangle$: A Gr\"{o}ebner basis for an ideal in $R$, where $m\in \mathbb{Z}_{>0}$.
\item $LT(f)$: The leading term of $f$ with respect to the order $>$.
\item $LCM(i,j)$: The lowest common multiple of the two polynomial $f_i$ and $f_j$.
\item $lcm(i,j)$: The lowest common multiple of the leading term for two polynomial $f_i$ and $f_j$.
\end{itemize}

\subsection{A Warming-up Example}
Now let us consider an ideal $\langle f_1=xy^3+z, f_2=x^2y^2+z^4,f_3=z^5\rangle$. There are  three principle syzygies $(f_2, -f_1,0), (f_3,0,-f_1),(0,f_3,-f_2)$. A general syzygy has the form of $(f_{c1}, f_{c2}, f_{c3})$. The eliminating between pairs can be done in a decreasing monomial order. The highest critical pair cancelation happens between $\{f_1, f_2\}, \{f_1, f_3\}, \{f_2, f_3\}$.We use $\{f_i,f_j\}$ ($\{i,j\}$ for simplicity) to label the critical pair. For  $\{f_1, f_2\}$, the cancelation between the highest terms leads to $LT(f_{c1})=m_1 x, LT(f_{c2})=m_1 (-y)$, where $m_1$ is an monomial in the polynomial ring $K[x,y,z]$. For the pairs $\{f_1, f_3\}$ and $\{f_2, f_3\}$, the leading terms are $LT(f_{c1})=m_1 z^5, LT(f_{c3})=m_1 (-x y^3)$ and $LT(f_{c2})=m_1 z^5, LT(f_{c2})=m_1 (-x^2 y^2)$. For the last two cases, the leading terms of the syzygies can be simplified by the principle syzygy $(f_3,0,-f_1),(0,f_3,-f_2)$. Hence such pair  $\{f_1, f_3\}$ and $\{f_2, f_3\}$ can be reduced by other syzygies. We will call a syzygy's leading monomial of the first term a barrier. A new syzygy can be added to the syzygy set if and only if its product factor can not be divided by any existing barrier.

We denote the syzygies and the barriers as
$$
\bordermatrix{
\text{index}&\text{Effect}&f_1&f_2&f_3\cr
\mathbf{S}_1&\text{T}&\underline{x^2 y^2}+z^4 & -x y^3-z & 0 \cr
\mathbf{S}_2&\text{T}& \underline{z^5} & 0 & -x y^3-z \cr
\mathbf{S}_3&\text{T} &0 & \underline{z^5} & -x^2 y^2-z^4  \cr
}
$$
where the barriers are underlined. We use $\text{T}$ to mark an irreducible syzygy for the original m-tuple $f_i$ at this step. The syzygy set of m-tuple $f_i$ denoted as $\mathcal{S}=\{\mathbf{S}_1,\mathbf{S}_2,\mathbf{S}_3\}$ is called the Top-syzygy set .
Now we can proceed to the next step. Following the rules described above, the only allowed pairs are those with product factors that can not be divided by the barriers. Here only the pair $\{f_1, f_2\}$ with the coefficients $(x, -y)$ is allowed.

After processing this pair, the ideal can be written as
$$\langle xy^3+z, x^2y^2+z^4,z^5, f_4=-yz^4+xz\rangle.$$
and there is one more syzygy $\mathbf{S}_{1,2}=(x,-y,0,-1)$ to be added to complete the syzygy set. This extra relation is an inheritance of the cancelation that takes place in the last step. Then the syzygies become:
\begin{figure}
  \centering
  \includegraphics[width=18cm]{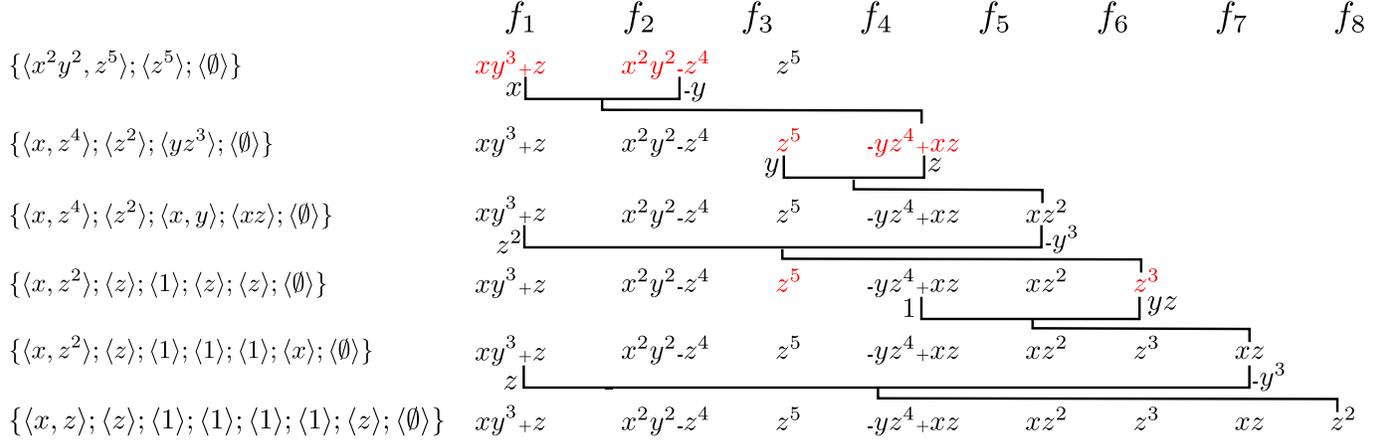}\\
  \caption{Illustrations on the warming-up example.}\label{example}
\end{figure}
$$
\bordermatrix{
\text{Index}&\text{Effect}&f_1&f_2&f_3&f_4\cr
\mathbf{S}_1&\text{T}& \underline{x^2 y^2}+z^4 & -x y^3-z & 0 & 0 \cr
\mathbf{S}_2&\text{T}& \underline{z^5} & 0 & -x y^3-z & 0 \cr
\mathbf{S}_3&\text{T}& 0 & \underline{z^5} & -x^2 y^2-z^4 & 0 \cr
\mathbf{S}_{1,2}&&\underline{x} & -y & 0 & -1 \cr
}
$$

At this moment, we should make sure that the barriers can't be divide by each other by performing linear transformations. First, add the product of $\mathbf{S}_{1,2}$ and $-xy^2$ to $\mathbf{S}_1$ ($\mathbf{S}_{1,2}$ and $\mathbf{S}_1$ represent the fourth and the first syzygy in the syzygy matrix, i.e. the fourth and the first row in the syzygy matrix, respectively). Now we have
$$
\left(
\begin{array}{cccc}
 \underline{z^4} & -z & 0 & xy^2 \\
 \underline{z^5} & 0 & -x y^3-z & 0 \\
 0 & \underline{z^5} & -x^2 y^2-z^4 & 0 \\
 \underline{x} & -y & 0 & -1 \\
\end{array}
\right).
$$
Next, we see that the first term in $\mathbf{S}_2$ can be divided by the first term in $\mathbf{S}_1$. So we multiply $\mathbf{S}_1$ with $-z$ and add it to $\mathbf{S}_2$; and the syzygy matrix becomes
$$
\left(
\begin{array}{cccc}
 \underline{z^4} & -z & 0 & xy^2 \\
 0 & \underline{z^2} & -x y^3-z & -xy^2z \\
 0 & \underline{z^5} & -x^2 y^2-z^4 & 0 \\
 \underline{x} & -y & 0 & -1 \\
\end{array}
\right).
$$
Likewise, we further simplify $\mathbf{S}_3$ using $\mathbf{S}_2$. The matrix then reads,
$$
\left(
\begin{array}{cccc}
 \underline{z^4} & -z & 0 & xy^2 \\
 0 & \underline{z^2} & -x y^3-z & -xy^2z \\
 0 & 0 & \underline{xy^3z^3}-x^2 y^2 & xy^2z^4 \\
 \underline{x} & -y & 0 & -1 \\
\end{array}
\right).
$$

Now there are three new two-pair syzygies, i.e. those containing the new polynomial $-yz^4+xz$. Again by the rules above, the only permitted pair is $\{f_3, f_4\}$. Here the relation is $(0,0,f_4/z, -z^4)$. From these two-pair relations the syzygies with barriers are updated as
$$
\left(
\begin{array}{cccc}
 \underline{z^4} & -z & 0 & xy^2 \\
 0 & \underline{z^2} & -x y^3-z & -xy^2z \\
 0 & 0 & \underline{xy^3z^3}-x^2 y^2 & xy^2z^4 \\
 \underline{x} & -y & 0 & -1 \\
 0 & 0 & \underline{-y z^3}+x & -z^4 \\
\end{array}
\right).
$$
 We observe that $\mathbf{S}_3$ can be further reduced by $\mathbf{S}_4$ to zero. Thus $\mathbf{S}_3$ should be removed from the syzygy matrix and the Top-syzygy set. The syzygy matrix becomes
$$
\bordermatrix{
\text{Index}&\text{Effect}&f_1&f_2&f_3&f_4\cr
\mathbf{S}_1&\text{T}& \underline{z^4} & -z & 0 & xy^2 \cr
\mathbf{S}_2& \text{T }&0 & \underline{z^2} & -x y^3-z & -xy^2z \cr
  \mathbf{S}_{1,2} & & \underline{x} & -y & 0 & -1 \cr
\mathbf{S}_4&\text{T}& 0 & 0 & \underline{-y z^3}+x & -z^4 \cr
}
$$
and the Top-syzygy set becomes $\mathcal{S}=\{\mathbf{S}_1, \mathbf{S}_2, \mathbf{S}_4\}$.

We can continue to add a new polynomial to the ideal, namely the polynomial generated by the pair $\{f_3, f_4\}$ with the coefficients $(y, z)$, and the ideal is updated to
$$\langle xy^3+z, x^2y^2+z^4,z^5, -yz^4+xz,f_5=xz^2\rangle.$$
We need a new syzygy that generates $f_5$ and the non-principle syzygy $\{f_4, f_5\}$. The updated syzygy matrix becomes:
$$
\bordermatrix{
\text{Index}&\text{Effect}&f_1&f_2&f_3&f_4&f_5\cr
\mathbf{S}_1&\text{T}& \underline{z^4} & -z & 0 & xy^2 &0\cr
\mathbf{S}_2& \text{T}&0 & \underline{z^2} & -z & 0 & -x y^2 \cr
\mathbf{S}_{1,2}&& \underline{x} & -y & 0 & -1 &0\cr
\mathbf{S}_4&\text{T}& 0 & 0 & \underline{x} & 0 & -z^3 \cr
\mathbf{S}_{3,4}&& 0 & 0 & \underline{y} & z&-1 \cr
\mathbf{S}_5&\text{F}& 0 & 0 &0& \underline{x z} & y z^3-x \cr
}.
$$
Here the new rank-2 syzygy $\mathbf{S}_5$  denoted by $F$ can generate a new barrier which can not be reduced to zero. However if we keep track of it back to the level above, we find this syzygy can be reduced to zero. Hence it can not contribute a new irreducible syzygy to the Top-syzygy set $\mathcal{S}$. To keep the Top-syzygies irreducibility, we introduce a rule that if a T-type syzygy is reduced to zero only by the F-type syzygies, we still keep this T-type syzygy in the Top-syzygy set $\mathcal{S}$.

The next permitted pair is $\{f_1, f_5\}$ with the coefficients $(z^2, -y^3)$, and it generates a new polynomial $f_6=z^3$. The syzygy $\mathbf{S}_{1,5}$, together with non-principle syzygy $\mathbf{S}_6$--$\{f_3, f_6\}$ with the coefficients $(1, -z^2)$ and $\mathbf{S}_7$--$\{f_5, f_6\}$ with the coefficients $(z, -x)$, update the syzygy matrix.  The syzygy $\mathbf{S}_6$ is T-type while $\mathbf{S}_7$ is F-type. Thus we can add $\mathbf{S}_6$ to the Top-syzygy set. After rewriting, the syzygy matrix refreshes to
$$
\bordermatrix{
\text{Index}&\text{Effect}&f_1&f_2&f_3&f_4&f_5&f_6\cr
\mathbf{S}_1&\text{T}& 0 & \underline{-z} & 0 & xy^2 & 0 & z^2+xy^3z \cr
\mathbf{S}_{1,2}& & \underline{x} & -y & 0 & -1 & 0 & 0 \cr
\mathbf{S}_{3,4}& & 0 & 0 & 0 &  \underline{z} &-1 & yz^2 \cr
 \mathbf{S}_{1,5}&&  \underline{z^2} & 0 & 0 & 0 & -y^3 & -1 \cr
\mathbf{S}_6&\text{T}& 0 & 0 &  \underline{1} & 0 & 0 & -z^2 \cr
\mathbf{S}_4&\text{T}& 0 & 0 & 0 & 0 &  \underline{z} & -x \cr
}.
$$
$\mathbf{S}_2$ is reduced to zero by $\mathbf{S}_6$ and some other non F-type syzygies. Hence it should be removed. $\mathbf{S}_4$ is reduced to zero only by F-type syzygies. This indicates that $\mathbf{S}_4$ and $\mathbf{S}_7$ is equivalent. We replace the elements in $\mathbf{S}_4$ by those in $\mathbf{S}_7$ and leave the index untouched. Hence the Top-syzygy set is $\mathcal{S}=\{\mathbf{S}_1, \mathbf{S}_4, \mathbf{S}_6\}$ at this step.

The next permitted pair is $\{f_4, f_6\}$ with the coefficients $(1, yz)$, and it generates a new polynomial $f_7 = xz$. After rewriting, the syzygy matrix refreshes to
$$
\bordermatrix{
\text{Index}&\text{Effect}&f_1&f_2&f_3&f_4&f_5&f_6& f_7\cr
\mathbf{S}_1&\text{T}& 0 & \underline{-z} & 0 & 0 & 0 & z^2 &xy^2 \cr
\mathbf{S}_{1,2}& & \underline{x} & -y & 0 & 0 & 0 & yz & -1 \cr
\mathbf{S}_{3,4}&& 0 & 0 & 0 &  0 &  \underline{-1} & 0 & z \cr
 \mathbf{S}_{1,5}&&  \underline{z^2} & 0 & 0 & 0 & 0 & -1&-y^3z \cr
\mathbf{S}_6&\text{T}& 0 & 0 &  \underline{1} & 0 & 0 & -z^2 & 0 \cr
\mathbf{S}_4&\text{T}& 0 & 0 & 0 & 0 &  0 &  \underline{x} & -z^2 \cr
\mathbf{S}_{4,6}& & 0 & 0 & 0 &  \underline{1} & 0 & yz & -1 \cr
}
$$
The next permitted pair is $\{f_1, f_7\}$ with the coefficients $(z, -y^3)$, and it generates a new polynomial $f_8 = z^2$. Taken into consideration the non-principle syzygy $\{f_7, f_8\}$ with the coefficients $(z, -x)$, the syzygy matrix can be reduced as
$$
\bordermatrix{
\text{Index}&\text{Effect}&f_1&f_2&f_3&f_4&f_5&f_6& f_7&f_8\cr
\mathbf{S}_1&\text{T}& 0 & \underline{-z} & 0 &0 & 0 & 0&xy^2 &z^3\cr
\mathbf{S}_{1,2}& & \underline{x} & -y & 0 & 0 & 0 & 0 & -1&yz^2 \cr
\mathbf{S}_{3,4}&& 0 & 0 & 0 &  0 &  \underline{-1} & 0 & z &0\cr
 \mathbf{S}_{1,5}&& 0& 0 & 0 & 0 &0 &\underline{-1}&0&z \cr
\mathbf{S}_6&\text{T}& 0 & 0 &  \underline{1} & 0 & 0 & 0 & 0 &-z^3\cr
\mathbf{S}_{4,6}& & 0 & 0 & 0 &  \underline{1} & 0 & 0 & -1&yz^2 \cr
\mathbf{S}_{1,7}& & \underline{z} & 0 & 0 & 0 & 0 & 0 & -y^3 & -1 \cr
\mathbf{S}_{4}&\text{T}& 0 & 0 & 0 & 0 &  0 &  0 & \underline{z} & -x \cr
}.
$$

There is no more permitted pair, so the program halts. The final ideal is
$$\langle xy^3+z, x^2y^2+z^4,z^5, -yz^4+xz,f_5=xz^2,f_6=z^3,f_7=xz,f_8=z^2  \rangle.$$
In processing the last two critical pair, the Top-syzygy set is invariant. Now we obtain the syzygy $\mathcal{S}=\{\mathbf{S}_1,\mathbf{S}_4,\mathbf{S}_6\}$ of the original 3-tuple ideal generators,
$$
\left(
\begin{array}{ccc}
 z^4+x^2 y^2 & -x y^3-z & 0 \\
 x z^4 & -y z^4 & y z^3-x \\
 x y^3 z^3-z^4 & -y^4 z^3 &y^4 z^2+1
\end{array}
\right).
$$
The simplest form of the syzygies is obtained by pairing the elements in red in Fig. \ref{example}. The syzygy module obtained here is the same as that obtained by \textit{{\sc Singular}} \cite{Singular}. But the number of the generators of the module is smaller than that in \textit{{\sc Singular}}. According to Fig. \ref{example} and the analysis in \cite{math13}, the polynomials in the final step whose barriers are not $1$ just form the Gr\"{o}bner basis of the ideal.

\subsection{Syzygies Of General Ideals}
In this section we give a general description of our method. We arrange the generators of the ideal in an order $LT(f_i)<LT(f_j)$ for $i<j$. We also assume all the monomials in $f_j$ can not be divided by the $LT(f_i)$ for all $i<j$. This is easy to be realized by dividing each $f_j$ by all $f_i$ with $i<j$, which is called Top-reduction. Our method is  to generate the syzygies inductively from a series of critical pairs, as shown in Fig. \ref{Fig-first}.
\begin{figure}
  \centering
  \includegraphics[width=12cm]{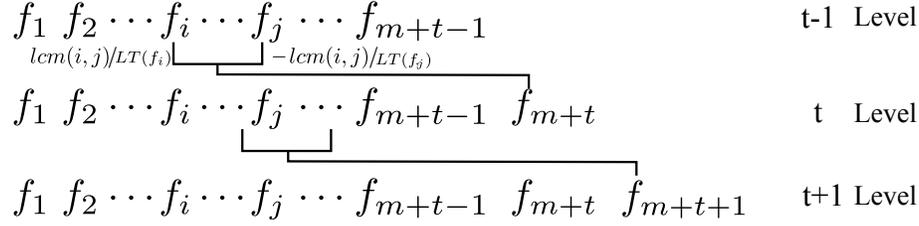}\\
  \caption{Performing syzygies by a series of critical pairs.}  \label{Fig-first}
\end{figure}
 We divide the generators into two levels, namely, before and after dealing with the critical pair, as shown in the first two rows in Fig. \ref{Fig-first}. In the up-row, there are already some obvious syzygies. If a pair induces a syzygy which can be simplified by the already known syzygies, then this pair is not permitted. This condition can be used to set up some barriers for the critical pairs. Among the allowed pairs, we prefer to deal with the one that generates the S-polynomial with the lowest leading terms in the order $>$. To summarize, in each step, we need to implement the following procedures
\begin{itemize}
\item Get all the trivial syzygies
\item Reduce the syzygies
\item Build up the barriers
\item Find a critical pair and generate an S-polynomial
\item Top-reduction on the S-polynomial
\end{itemize}

\paragraph{Get all the trivial syzygies} Now we set up the known syzygies on level $t$ as shown in Fig. \ref{Fig-first}. There are four kinds of syzygies on level $t$ which can be directly read off
\begin{itemize}
\item Syzygies inherited from the syzygies on level $t-1$
\item Syzygy induced by the critical pair on level $t-1$,  $(\cdots { lcm_{i,j}\over LT(f_j)}\cdots{lcm_{i,j}\over LT(f_i)}\cdots0, -1)$. We denote such syzygies as $\mathbf{S}_{i,j}$
\item Principle syzygies $\mathbf{S}^p$ with $f_{m+t}$
\item Non-principle rank-2 syzygies with $f_{m+t}$. On level $t$,  they are just \\
$(0,\cdots, {f_{m+t}\over c_{i,m+t}},\cdots, {-f_i\over c_{i,m+t} } )$, where  $c_{i,m+t}$ is the maximal common factor of $f_i$ and $f_{m+t}$. We use $\mathbf{S}^e_{i,m+t}$ to denote such a syzygy
\end{itemize}
For the last kind of syzygies, they are permitted only when they are not blocked by the former syzygies. For setting up the barriers, the principle syzygies with $f_{m+t}$ are not necessary since they always generate some abandoned syzygies. To see this, we consider the principle syzygy between $f_1$ and $f_{m+t}$ without loss of generality 
$$\mathbf{S}^p_{1,m+t}=\left(f_{m+t},\cdots, 0_i,\cdots, 0_j,\cdots, -f_1\right).$$
 The syzygies inherited from the level $(t-1)$ contain 
 $$\mathbf{S}^p_{1,i}=\left(f_i,\cdots, -f_1,\cdots, 0_j,\cdots, 0_{m+r}\right)$$ and $$\mathbf{S}^p_{1,j}=\left(f_j,\cdots, 0_i,\cdots, -f_1,\cdots, 0_{m+r}\right).$$ According to the syzygy induced by the critical pair $$\mathbf{S}_{i,j}=(0_1,\cdots ,{ lcm_{i,j}\over LT(f_j)},\cdots, -{lcm_{i,j}\over LT(f_i)}, \cdots, -1),$$
  it is easy to see that $f_{m+t}={ lcm_{i,j}\over LT(f_j)} f_i- {lcm_{i,j}\over LT(f_i)} f_j$ and  $$\mathbf{S}^p_{1,m+t}={lcm_{i,j}\over LT(f_j)}\mathbf{S}^p_{1,i}- {lcm_{i,j}\over LT(f_i)}\mathbf{S}^p_{1,j}+f_1 \mathbf{S}_{i,j}.$$ Back to the level $(t-1)$, $\mathbf{S}_{i,j}$ is vanishing and the corresponding syzygy of   $\mathbf{S}^p_{1,m+t}$ on level $(t-1)$ is $$\mathbf{S}^{p,t-1}_{1,m+t}={lcm_{i,j}\over LT(f_j)}\mathbf{S}^p_{1,i}-{lcm_{i,j}\over LT(f_i)}\mathbf{S}^p_{1,j}.$$  Hence the syzygy $\mathbf{S}^p_{1,m+t}$ is reduced to zero directly on level $t-1$ by $\mathbf{S}^p_{1,i}$ and $\mathbf{S}^p_{1,j}$. Obviously $\mathbf{S}^p_{1,m+t}$ can be deleted without affecting the completeness of the syzygy module.  Furthermore, for the justification of irreducibility, if a syzygy  $\mathbf{S}$ is reduced to zero by $\mathbf{S}^p_{1,m+t}$ on level $t$, it is  also reduced  to  zero by $\mathbf{S}^p_{1,i}$ and $\mathbf{S}^p_{1,j}$ on level $(t-1)$. Hence we do not include any principle syzygy with $f_{m+t}$.

\paragraph{Syzygies, Rewriting and Barriers}
The known syzygies can usually be reduced with each other and this process provides more barriers for upcoming critical pairs. We first reduce each syzygy using $\mathbf{S}_{i,j}$. Then we set up the initial barrier $\langle B^0_{f_i}\rangle$ for each $f_i$. If an pair is blocked by in $\langle B^0_{f_i}\rangle$, then the pair multiplying an monomial is also blocked. The barrier $\langle B^0_{f_i}\rangle$ is an ideal in the monomial ring.  For each syzygy other than $\mathbf{S}_{i,j}$, we select the leading term of the first non-vanishing polynomial in the syzygy in order  $\vec{e}_1\succ\vec{e}_2\succ\cdots\succ\vec{e}_{m+t}$. Then we add this leading term into the generating set for $\langle B^0_{f_i}\rangle$. For each critical pair, it induces a syzygy $\mathbf{S}_{i,j}$ which forbids a particular critical pair $\{i, j\}$ for upcoming syzygies. It also adds a new generator  ${LCM(i,j)\over LT(f_i)}$ to the ideal $\langle B^0_{f_i}\rangle$ if $\mathbf{S}_{i,j}[k]=0$ for all $k<i$. We use $\mathbf{S}_{i,j}[k]$ to denote the k-column value of $\mathbf{S}_{i,j}$.

For the syzygies which can create  generators for $\langle B^0_{f_i}\rangle$, they usually can be reduced using each other. To see how this happens, we choose two syzygies $\mathbf{S}_1, \mathbf{S}_2$ which create the generators $m_1, m_2$ for $\langle B^0_{f_i}\rangle$. Then the form of the two syzygies are 
\begin{eqnarray*}
\mathbf{S}_1&=&\{0_1,\cdots,0_{i-1}, m_1+h_1, \cdots\}\\
\mathbf{S}_2&=&\{0_1,\cdots,0_{i-1}, m_2+h_2, \cdots\},
\end{eqnarray*}
where $m_1, m_2$ are monomials and $h_1,h_2$ are polynomials. 
If $m_1=t \times m_2$, where $t$ is an element in the monomial ring, then $\mathbf{S}_1$ can be reduced by $\mathbf{S}_2$ as  $$\mathbf{S}'_1=\mathbf{S}_1-t \mathbf{S}_2=\{0_1,\cdots,0_{i-1}, h_1- t h_2, \cdots\}.$$ The syzygy $\mathbf{S}'_1$ will create a generator $m_3=LT(h_1- t h_2)$ for  the barrier ideal of $f_i$. Then the barrier ideal is $\langle B^1_{f_i}\rangle=\langle m_3, m_1, m_2,\cdots\rangle=\langle m_3, m_2,\cdots\rangle$, where $\cdots$ denote the other generators  of $\langle B^0_{f_i}\rangle$ from other syzygies. The involution of the barrier ideal under rewriting have $\langle B^0_{f_i}\rangle \subseteq \langle B^1_{f_i}\rangle$.  When  $\langle B^0_{f_i}\rangle= \langle B^1_{f_i}\rangle$, this indicates that $\mathbf{S}'_1$ can also be reduced by other syzygies. The rewriting will stop only when $\langle B^0_{f_i}\rangle \subset \langle B^1_{f_i}\rangle$ or $m_3$ is reduced to zero. If $m_3$ is reduced to zero, the final $\mathbf{S}'_1$ after several rewriting steps will add a generator to the barrier ideal $\langle B^0_{f_{i+1}}\rangle$. The generator-creating rewriting process of syzygies can be done recursively. Hence we conclude that the rewriting can enlarge the barrier ideal or leave the ideal unchanged. And after the steps of rewriting for all columns $f_{i}$, we get the maximal barrier ideal $\langle B_{f_{i}}\rangle$ for each $f_{i}$. 

If there are some $\mathbf{S}_{i,j}[k]\neq 0$ for $k<i$, this kind of syzygies do not get reduced in the steps above.  They only forbid a particular critical pair $\{i, j\}$ to  upcoming syzygies. To reduce them, we can use the syzygy $\mathbf{S}^e$ when $\mathbf{S}[i]$ can be divided by $LT(\mathbf{S}^e[i])$ and $\mathbf{S}^e[k<i]=0$, or when $\mathbf{S}[j]$ can be divided by $LT(\mathbf{S}^e[j])$ and $\mathbf{S}^e[k<j]=0$. In either case, the critical pair $\{i, j\}$ has been blocked by the barrier ideal generator $LT(\mathbf{S}^e[i])$ or $LT(\mathbf{S}^e[j])$.  After rewriting, $\mathbf{S}_{i,j}$ becomes  $\mathbf{S}'_{i,j}=\mathbf{S}_{i,j}-m_1 \mathbf{S}^e$, where $m_1={\mathbf{S}[i]\over LT(\mathbf{S}^e[i])}, \text{or} {\mathbf{S}[j]\over LT(\mathbf{S}^e[j])}$. Furthermore $\mathbf{S}'_{i,j}$ does not contain the critical pair. $\mathbf{S}'_{i,j}$ also creates a generator for the barrier ideal $\langle B_{f_{i}}\rangle$ and return to the rewriting steps for the involution of barrier ideals. 

 \paragraph{Irreducibility among critical pairs}
 In addition to the rewriting of the syzygies, to maintain the irreducibility of the syzygies, it needs to be guaranteed that the critical pairs are irreducible among themselves. To justify the irreducibility of the critical pairs, only the leading terms of said pairs are relevant. The irreducibility is easy to be observed from the cell complex for the monomial ideas \cite{math12}. 
The cell complex can be set up on the present level $t$ or on level $(t-1)$ in advance.  As an example, we consider three leading terms $\{x^2y,xz^2, yz^3\}$ on level $t$. We first define the vertices of the simplex $\bigtriangleup$ to be the monomial elements as shown in Fig. \ref{Fig-TriRelation}.  
 \begin{figure}
  \centering
  \includegraphics[width=3cm]{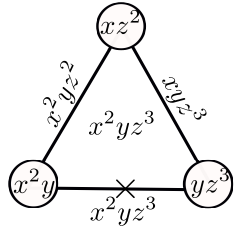}\\
  \caption{Example for Cell complex.}  \label{Fig-TriRelation}
\end{figure}
The edge of $\bigtriangleup$ is labelled by the least common multiple of the monomial element. And the face of $\bigtriangleup$ which the three edges form is labelled by the least common multiple of the labels of the edges. The edges $e_{12}, e_{23}, e_{13}$ imply three syzygies $(z^2, -xy,0)$,  $(0, yz,x)$,  $(z^3, 0,x^2)$. The least common multiple of $e_{13}$ is equal to the least common multiple of the face. Hence the syzygy corresponding to $e_{13}$ is reducible.  This rule holds generally and is very convenient for determining whether the syzygies generated by the critical pairs are reducible. In fact if the edges form a face, and the least common multiple of the face is equal to the least common multiple of a edge, the syzygy corresponding to the edge is reducible. 
 
The cell complex on level $(t-1)$ is shown in the following example. We set that the ideal generators include $f_{i}=x_4, f_{j}=x_2x_5x_6, f_{k}=x_1x_2x_3+x_1$ on level $(t-1)$. Then we perform the critical pair for $\{f_j, f_k\}$. A new generator $f_{k+1}=-x_1x_5x_6$ of the ideal is obtained on level $t$ as shown in  Fig. \ref{SAF}. 
 \begin{figure}
  \centering
  \includegraphics[width=12cm]{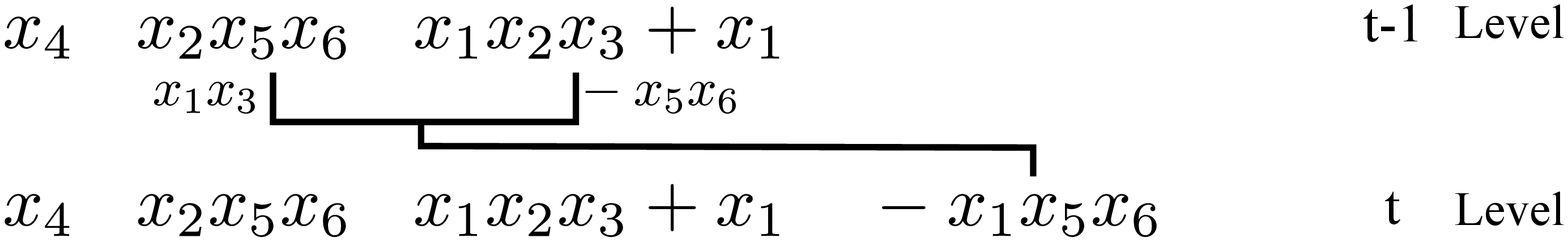}\\
  \caption{Performing a critical pair.}  \label{SAF}
\end{figure}
On level $(t-1)$, according to the cell complex, the three edges are denoted as $e_{ij}, e_{ik}, e_{jk}$. These edges are characterized by the least common multiples  $x_2x_4x_5x_6, x_1x_2x_3x_4, x_1x_2x_3x_5x_6$ for leading terms of $(f_i, f_j), (f_i, f_k), (f_j, f_k)$ respectively. The face $s_{ijk}$ formed by the three edges is characterized by the least common multiple $m_s=x_1x_2x_3x_4x_5x_6$ for  the leading terms of $(f_i, f_j,f_k)$.  None of the edge common multiple is equal to $m_s$. Hence the three critical pairs are irreducible among each other.  On level $t$, according to the cell complex, the irreducible edges are shown in Fig. \ref{SAS}. 
\begin{figure}
  \centering
  \includegraphics[width=4cm]{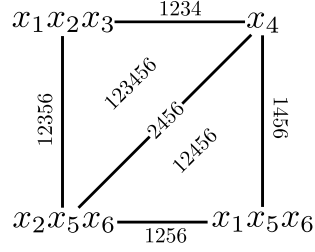}\\
  \caption{Cell Complex for the critical pairs, the symbol such as $1234$ denote the monomial $x_1x_2x_3x_4$ for convenience.}  \label{SAS}
\end{figure}
Now we discuss the critical pair for edge $e_{i,k+1}$. The syzygy corresponding to $e_{i,k+1}$ is $\mathbf{S}_{i,k+1}=\{x_1x_5x_6,0,0,x_4\}$ on level $t$. On level $(t-1)$, $\mathbf{S}_{i,k+1}$ become $\mathbf{S}'_{i,k+1}=\{x_1x_5x_6, x_1x_3x_4,-x_4x_5x_6\}$. This syzygy is the product of $x_4$ and the syzygy corresponding to $ e_{jk}$ $$x_4 (x_1 x_3  f_j - x_5 x_6  f_k)=-x_1 x_5 x_6 f_i. $$ This relation is reducible as the product of $x_4$ and the least common multiple $x_1x_2x_3x_5x_6$ for the edge $ e_{jk}$ divide the $m_s$ for the face $s_{ijk}$. Hence according to the cell complex, such syzygy can be reduced to 
\begin{eqnarray*}
\mathbf{S}_{i,j}&=&\{x_2x_5x_6,-x_4,0\}\\
\mathbf{S}_{i,k}&=&\{x_1x_2x_3+x_1,0,-x_4\}.
\end{eqnarray*}
  This is also direct to check $$\mathbf{S}'_{i,k+1}+{x_1x_3} \mathbf{S}_{i,j} -{x_5 x_6} \mathbf{S}_{i,k}=0.$$ Hence, $\mathbf{S}_{i,k+1}$ is  reducible on $\mathbf{S}_{i,j}$, $\mathbf{S}_{i,k}$, $\mathbf{S}_{j,k}$. We  just delete $\mathbf{S}_{i,k+1}$. In practice, the rules for forbidding such reducible critical pairs are summarized  as following: When performing a critical pair $\{i,j\}$ on level $(t-1)$ and generating new ideal generator $f_{m+i+1}$, choose generator $f_k$ on level $(t-1)$, such that $k\neq i, k\neq j$ and the corresponding S-polynomial $S_{k,i}, S_{k,j}$ can be reduced  to zero, then $\mathbf{S}_{m+t+1,k}$ is reducible by the syzygy module on level $(t-1)$ when ${lcm(LT(f_{m+t+1}), LT(f_k))\over LT(f_{m+t+1})}lcm(LT(f_i),LT(f_j))$ divide $lcm(LT(f_i),LT(f_j),LT(f_k))$. The proof of this statement is direct by the cell complex as shown in the above example.

\paragraph{Get the permitted critical pair}
After updating the  barriers, it is easy to observe which critical pair is permitted. If the monomials from a critical pair belong to $\langle B_{f_i}\rangle$, such pair is not permitted. And all the used critical pairs are not permitted in the following steps. Among all the allowed critical pairs, we choose a pair such that the leading term of S-polynomial is the lowest in the $>$ order. When performing a critical pair, a new polynomial is generated. In the following, we always assume a full top-reduction is performed on new polynomials.

\paragraph{Top-reduction}   
\begin{theorem}
For an ideal  $\langle f_1,f_2\cdots f_i\cdots f_j\cdots f_m\rangle$, if a critical pair $\{i,j\}$ generates a polynomial  $f_{m+1}$ with leading term being able to be divided by the leading terms of the  generators $f_s$, then the  $f_{m+1}$ is reduced to a polynomial $f'_{m+1}$ with the lower order leading term. All the irreducible syzygies will not refer to $f_{m+1}$, which means that we can replace $f_{m+1}$ with $f'_{m+1}$.
\end{theorem}
\begin{proof}
Without loss of generality, we suppose that $LT(f_{m+1})$ is divided by $LT(f_m)$, then any critical pair $\{i,m+1\}$ can be taken as the composition of $\{i,m\}$ and $\{m,m+1\}$, which is easy to see from the triangle diagram in Fig. \ref{Fig-TwoPairDecStep1}. The critical pairs, denoted by edges in Fig. \ref{Fig-TwoPairDecStep1}, are characterized by the minimal common factor of the leading terms of $f$. Three critical pair edges form a triangle. This triangle characterize the syzygy of the syzygy module of the three elements. If the face common factor equals an edge common factor, then according to \cite{math12} the syzygy corresponding to the edge is reducible.  Since all syzygies are decomposed as the compositions of the critical pairs, any  irreducible syzygy does not refer to $f_{m+1}$. We replace $f_{m+1}$ with $f'_{m+1}$, which is named as Top-reduction. The Top-reduction is done recursively and finally we have a fully reduced $f'_{m+1}$. If the Top-reduction gives us $f'_{m+1}=0$, this will induce a syzygy. This syzygy is a high rank syzygy. After adding this to the syzygy module, the critical pair can never be used in further steps.
\end{proof}
If a S-polynomial is reduce to zero under Top-reduction, then there is a higher rank trivial syzygy
$$\left(\begin{array}{cccccccccc}
 & f_1&\cdots & f_{i-1}& f_i                     & \cdots &f_{j-1}& f_j                     &\cdots  &f_m \cr
 & a_1&\cdots &a_{i-1} &{LCMij\over LT(f_i)}+a_i &\cdots&a_{j-1}  & {LCMij\over LT(f_j)}+a_j&\cdots  &a_m
\end{array}\right)$$
where $a_i$ is the coefficient of $f_i$ under Top-reduction for $S_{i,j}$. This  only happens in the beginning of the critical pair decomposition. Our algorithm is able to forbid such redundant pairs.
\begin{figure}
  \centering
  \includegraphics[width=12cm]{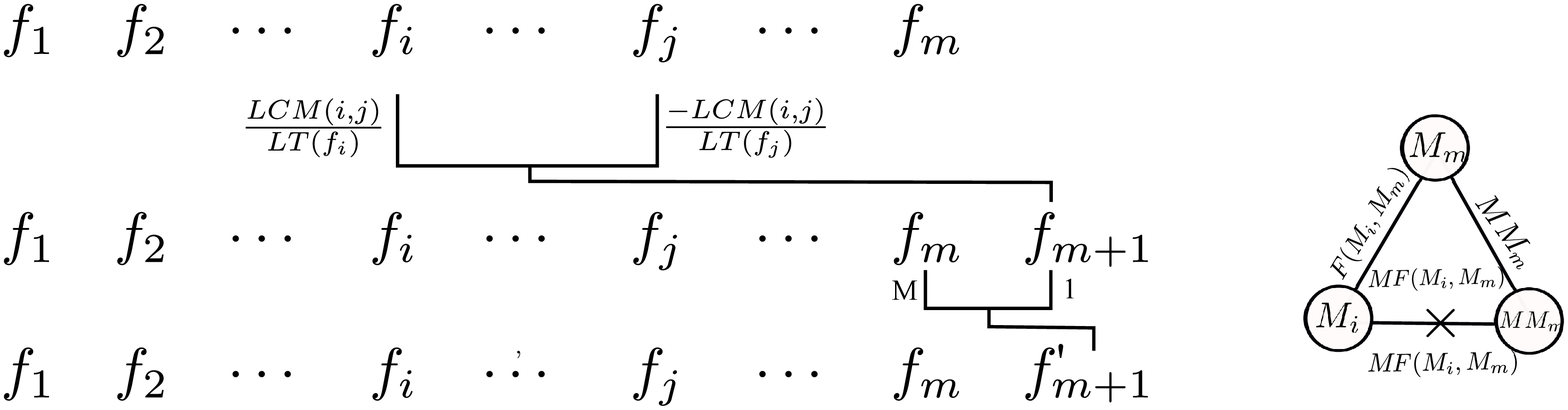}\\
  \caption{Triangle diagram for our theorem.}  \label{Fig-TwoPairDecStep1}
\end{figure}

\paragraph{End criterion }
The former procedures are  proceeded recursively and stop if all the critical pairs are not permitted. This means that all  other syzygies can be reduce to zero under the existing syzygies, and we finally obtain the  all the irreducible  syzygies. These syzygies form a  basis of the syzygy module. Moreover, we also get the Gr\"{o}bner basis of the ideal. The  Gr\"{o}bner basis is formed by the polynomials in final step whose barriers are not equal to $\langle 1\rangle$. This procedure is shown in flow chart Fig. \ref{Fig-C2Z}.
\begin{figure}
  \centering
  \includegraphics[width=8cm]{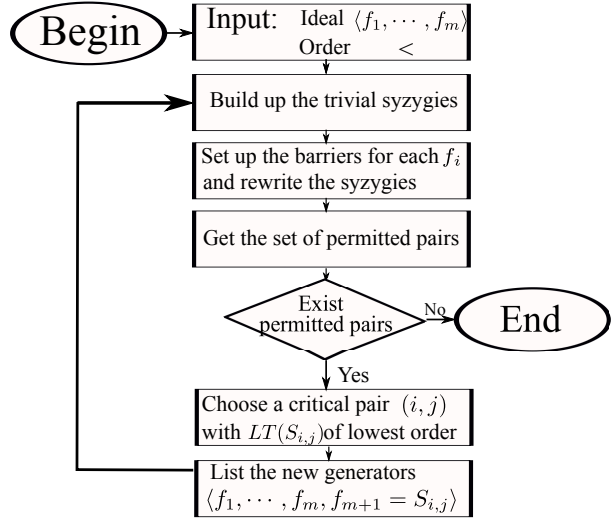}\\
  \caption{End criterion.}  \label{Fig-C2Z}
\end{figure}

\paragraph{Independence criterion  for the syzygy}
In performing each critical pair, we can keep the new added syzygy from being reduced by the old ones. This is only part of the story for the irreducibility among syzygies. To remove all the  reducible syzygies and obtain the locally minimal generating set of the syzygies, we need to deal with much more in each step. Let's first discuss relations between syzygies on level $(t-1)$ and those on level $t$.
\begin{itemize}
\item $\{\text{Syzygies on level $t-1$}\}\subset\{\text{Syzygies on level $t$}\}$.
\item A permitted rank-2 syzygy on level $t$ combining with the critical pair induced syzygy  will generate a higher rank syzygy on level $t-1$.
\end{itemize}
When we use the syzygies on level $t$ to set up the barriers, new permitted rank-2 syzygies can not be reduced on level $t$. If it is reduced to zero at  on level $t-1$, we denote this syzygy as $F$. It means that this syzygy generates a new barrier to the following steps, but a reducible syzygy on level $t-1$. In turn, this also generates a linearly reducible syzygy for the original m-tuple polynomial and we do not need to add it to the Top-syzygy set. If it can not be reduced to zero on level $t-1$, then we denote it a $T$ and add it to the Top-syzygy set.

We suppose that  the rewriting on level $t-1$ is performed and a syzygy is reduced to zero. If it is of $F$-type, then we just remove it. For $T$-type syzygy,  we need to be more careful. If it is reduced to zero all by  $T$-type and critical pair syzygies, then we remove it. If it is reduced to zero by a $F$-type syzygy and some other type of syzygies, then we replace it by this $F$-type syzygy in the Top-syzygy set. This is because the two syzygies are equivalent. In order to continue the  following steps, this $F$-type syzygy is used. However for the sake of obtaining the right and simple Top-syzygies, we use the $T$-type syzygy. Hence the replacing $F$-type syzygy also come from the eliminating pair of the $T$-type syzygy. Remark that if the original $T$-type syzygy is reduced to zero by several $F$-type syzygies, we should replace it syzygy by the last adding $F$-type syzygy. The irreducibility  of these syzygies is justified by rewriting of the syzygy module as is shown in Fig. \ref{Fig-OutputSyz}.
\begin{figure}
  \centering
  \includegraphics[width=15cm]{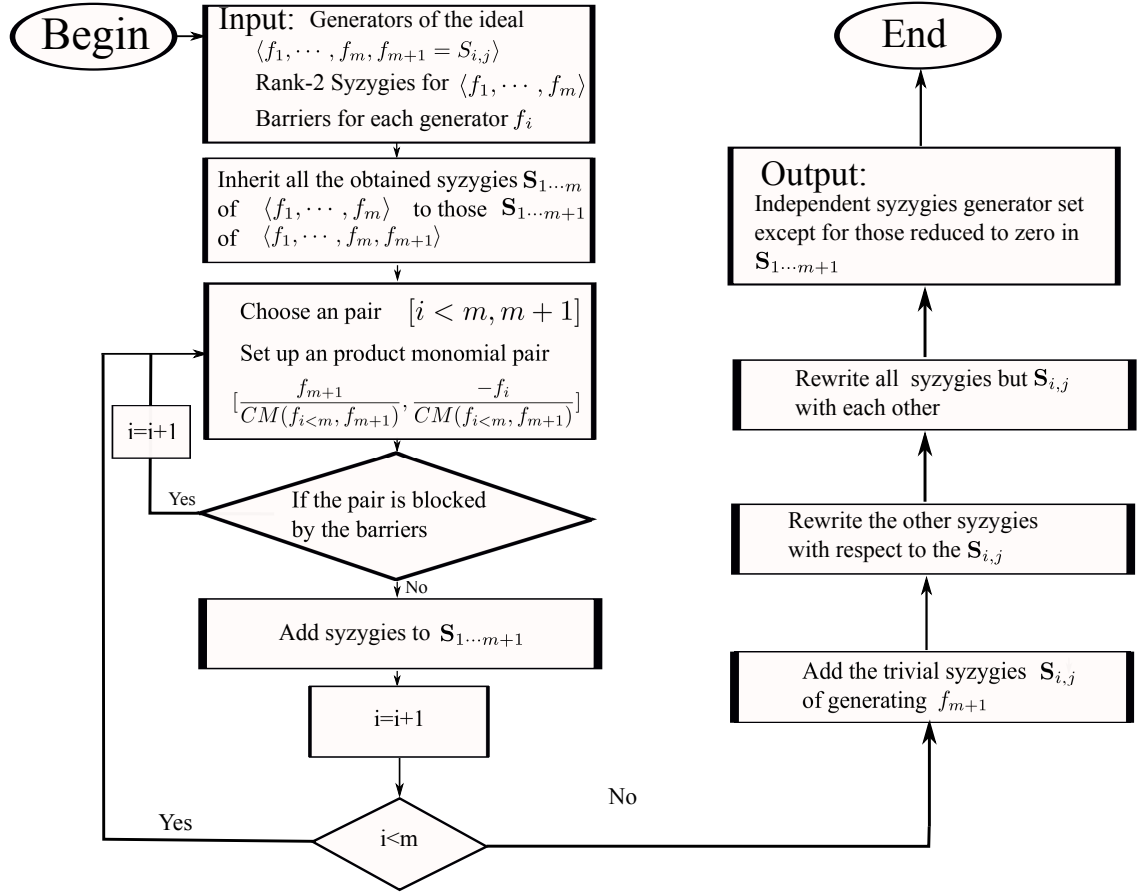}\\
  \caption{Top-reduction of the syzygy module.}  \label{Fig-OutputSyz}
\end{figure}
The above procedure is performed recursively. In principle, the reducibility among the syzygies can not be fully determined at each level. However, in performing the critical pairs one by one, if we keep to reduce all the syzygies at each level, and delete all the abandoned syzygies one by one, we can finally get the irreducible syzygies for zero level generators. This is easy to see. For convenience, we define the leading critical pair for each syzygy as following:  A pair with coefficients $(C^{S}_{i}, C^{S}_{j})$ of column $i$ and column $j$ in $\mathbf{S}$ is called the leading critical pair if $C^{S}_{i}LT(f_i)+C^{S}_{j}LT(f_j)=0$ and  $C^{S}_{i}LT(f_i)$ is of the largest monomial order comparing with other pairs,  where $C^{S}$ is the leading term of  $\mathbf{S}$.   We only need to verify that at the end of our procedure if $\mathbf{S}_t=c_1\mathbf{S}_1+c_2\mathbf{S}_2\cdots+c_K\mathbf{S}_K$, then $\mathbf{S}_t$ will be reduced to zero.  According to the equation, the leading critical pair of $\mathbf{S}_t$ should be reducible by the  leading critical pair of $\mathbf{S}_k$, $k\in{1\cdots K}$.  The leading critical pair in each $\mathbf{S}_k$  contains the index of the first non-zero column except for those critical pair induced syzygies.  Otherwise there will be more allowed critical pairs and the main loop can not break, thus the leading critical pair of $\mathbf{S}_t$ is able to be reduced by the syzygies  $\mathbf{S}_k$, $k\in{1\cdots K}$.

Our algorithm generates the complete syzygy module. In fact, we include all the possible syzygy in performing  each critical pair. We only  abandon  those reducible syzygies. After performing all the irreducible critical pairs, there is no other syzygy not in the module.  To prove this, we only need to verify that any syzygy $\mathbf{S}$ can be decomposed by  $\mathbf{S}_1, \mathbf{S}_2\cdots \mathbf{S}_K$, where $\mathbf{S}_{i\in [1, K]}$ is the syzygy module we obtained. For each syzygy of the ideal generators, the leading critical pair should be a summation of the multiples of some critical pairs $\{i_1,j_1\},\cdots, \{i_{n_1},j_{n_2}\}$.  If the critical pair does not generate a new ideal generator, this critical pair will generate a syzygy $\mathbf{S}_i$ which is in our syzygy module. Then we have $\mathbf{S}=\mathbf{S}'+\mathbf{S}_i$.  For $\mathbf{S}'$, the leading pair is of small order.  If the critical pair generates a new  ideal generator, the $\mathbf{S}$ become a syzygy $\mathbf{S}^1$ in next level. After transforming the leading critical pair, we find that the syzygy $\mathbf{S}$ either decreases the order or becomes  syzygy on next level.  Then we transform the leading critical pair step by step. Finally, the steps stop when $\mathbf{S}$ become  either zero or a syzygy on the final level. The ideal generators on the final level form a  Gr\"{o}bner basis. For the first case, it just means $\mathbf{S}$ lies in our syzygy module. For the last case, $\mathbf{S}$ is able to be decomposed as the summation of the syzygies induced by the critical pairs. At finial step, each critical pair is blocked by the barriers. Then the corresponding syzygies are all in the  syzygy module. Hence $\mathbf{S}$ belong to the syzygy module. 

Now, we get the full description of this new algorithm, we name it as $\text{C2Z}$ algorithm.  In the following, we apply our algorithm to several examples.
\paragraph{A short cut for regular sequence} We first discuss two-tuple regular sequence $(f_1, f_2)$. The principle syzygy  is $(f_2, -f_1)$. In our algorithm,  this can generate a barrier ideal $\langle LT(f_2)\rangle$ on $f_1$.  Then the critical pair $\{1,2\}$ is permitted when $LT(f_1)$ and $LT(f_2)$ have a proper greatest common factor.  Our algorithm will stop when we get a Gr\"{o}bner basis.  The  syzygy module of $(f_1, f_2)$ is induced by all the syzygies in the following levels. If a new syzygy on the following level  is of  T-type, then the corresponding Top-syzygy $(h_1, h_2)$ of $(f_1, f_2)$ is not blocked by the barrier ideal $\langle LT(f_2)\rangle$. Here $h_1, h_2$ are  polynomials. Hence $LT(f_2)$ is not a  divisor of $LT(h_1)$. But  $(f_1, f_2)$ is a  regular sequence. Then  $f_2$ should be a divisor of $h_1$. In turn, $LT(f_2)$ is a  divisor of $LT(h_1)$. Then we get the T-type syzygy does not exist. The syzygy module is generated  by the principle syzygy for regular sequence $(f_1, f_2)$. 

Likely, for $m$-tuple regular sequence  $(f_1, f_2,\cdots, f_m)$,  a syzygy arising on a following  level $t$ generate a Top-syzygy. If it is of T-type, the corresponding  Top-syzygy is of form $\mathbf{S}_t=(0_1, 0_2,\cdots, 0_i, h_{i+1}, h_{i+2}, \cdots, h_{m})$ in general. For regular sequence, we have $h_{i+1}\in \langle f_{i+2}, f_{i+3},\cdots, f_{m}\rangle$. Since we finally get the Gr\"{o}bner basis for the ideal.   $\mathbf{S}_t$ will be reduced by the principle syzygies on a level below $t$. No new syzygy is  added to the syzygy module for the $m$-tuple regular sequence  $(f_1, f_2,\cdots, f_m)$.  Furthermore, the principle syzygies are  possible to be reduced among each other. This is easy to be realized by the cell complex composed by the leading terms of entry in $(f_1, f_2,\cdots, f_m)$. Such kind of reduction is also included in our algorithm.  

To test the program, we need to compare with the popular software \textit{{\sc Singular}} \cite{Singular}. By \textit{{\sc Singular}}, the local minimal generating set of the syzygy module is obtained according to the Gr\"{o}bner basis of the syzygies of the syzygy module. In the Gr\"{o}bner basis, a syzygy of the syzygy module has
$$(h_1,\cdots, h_{i-1}, h_{i}, h_{i+1}, \cdots, h_K)\cdot (\mathbf{S}_1,\cdots,\mathbf{S}_{i-1},\mathbf{S}_i,\mathbf{S}_{i+1}, \cdots,\mathbf{S}_K)=0,$$ 
where $K$ is an integral number denoting the generator number of the syzygy module.  If one of the entry of the syzygy $h_i\in \mathbb{Z}$,  then $\mathbf{S}_i$ is reducible.  Let's consider the following ideals with  integral coefficients 
\begin{eqnarray*}
I_1&=&\left\langle x_2 + x_2^2 + x_1 x_3,\ x_1 + 3 x_1^2,\ x_1^2, x_1 x_2 + x_2^2 + x_3 + x_3^2,\ 
 x_2^2 + x_1 x_3 + x_3^2\right\rangle,\\
I_2&=&\left\langle x_3 + x_1 x_3 + x_2 x_3,\ x_1 + x_1 x_2 + 2 x_3^2,\ x_1 x_2 + x_3 + x_1 x_3 + x_2 x_3,\ 
x_1 x_2 + x_2^2 + x_3 + x_2 x_3,\right.\\
&&\left.x_1 x_2 + x_2^2 + x_1 x_3 + x_2 x_3, x_2^2 + 2 x_1 x_3 \right\rangle,\\
I_3&=&\left\langle x_3 + x_4 + x_1 x_4 + x_3 x_4,\ x_1 + x_1 x_4,\ x_2^2 + x_3 + x_3^2 \right\rangle,\\
I_4&=&\left\langle x_1 + 2 x_2 + x_1 x_4,\ x_1 x_2 + x_1 x_3 + x_4 + x_2 x_4,\ 
x_2 x_3 + x_3^2,\ x_1 + x_2 + x_4 + x_1 x_4,\ x_1 x_2 \right\rangle,\\
I_5&=&\left\langle x_1 x_2 x_3 + x_2 x_4^2,\ 
x_1 x_4 + x_4^2,\ x_4^3,\ x_4 + x_1 x_2 x_4 + x_3^2 x_4, x_3^3 + x_1 x_2 x_4,\ x_2 x_4^2 \right\rangle,\\
I_{6}&=&\left\langle x_1^2 x_3,x_1 x_2^2 x_3^2+x_1 x_2 x_3^3,x_1 x_2 x_3^4+x_3\right\rangle,\\
I_{7}&=&\left\langle x_1^6+x_1 x_2^3 x_3+x_1 x_2 x_3^4+x_2,x_1^4 x_3^2+x_1^3 x_2^2+x_1^2 x_2^3 x_3,x_1+x_2^3 x_3^2,x_1 x_2^2 x_3^3,x_3^6,x_1 x_2 x_3^2+x_3^3\right\rangle,\\
I_{8}&=&\left\langle x_1^3 x_2 x_3^3+x_1^3+x_1^2 x_2^2,x_1^6 x_3+x_1 x_2^2 x_3+x_1 x_2 x_3^2+x_1,x_3^3,x_1^3 x_3^2+x_1^3+x_3,x_1 x_3^4+x_2^2 x_3^3\right\rangle,\\
I_{9}&=&\left\langle x_1^4 x_2 x_3^3+x_1^2 x_2 x_3^5+x_1 x_2^3 x_3^3,x_1^2 x_2^4 x_3+x_2^6,x_1^3 x_2^3 x_3^2+x_1^2 x_2+x_2 x_3^2,x_1^3 x_2^2 x_3^2,x_1^2 x_2^2+x_1\right\rangle.\\
\end{eqnarray*}
We compare the our results  and the timing of this algorithm with \textit{{\sc Singular}}  on a computer with CPU 2.4GHz and RAM 8G as shown in following
$$
\begin{array}{|c|c|c|c|c|c|c|}\hline
 & T_1/s & n_1 & T_2/s & n_2&T_3/s& n_3\\ \hline
 I_1& 0.04 & 5 & 0.03& 6 & 0.29&5\\ \hline
 I_2& 0.05 & 6  &0.05& 7 & 0.57&6\\ \hline
 I_3& 0.09 & 3 & 0.23 & 4 &0.38&3\\ \hline
 I_4& 0.27 & 7  & 0.14 & 10 & 1.2&7\\ \hline
 I_5& 0.32 & 6 &  0.19 & 9 &1.0 &6\\ \hline
 I_{6}& 0.09 & 2  &0.14& 3 & 0.28&2\\ \hline
 I_{7}& 0.88 & 8 & 0.15 & 12 & ?&?\\ \hline
 I_{8}& 1.59 & 11  & 0.26 & 17 & 1.61&11\\ \hline
 I_{9}& 0.53 & 6 &  0.15 & 15 & ? &?\\ \hline
 \end{array},
$$
where the $T_1, n_1$ is the timing and number of syzygies of our program,  $T_2, n_2$ is the timing and number of syzygies  in \textit{{\sc Singular}}, $T_3, n_3$ is the timing and number for getting the local minimal generating set of syzygies in \textit{{\sc Singular}}, the symbol ``$?$" denotes that the time is longer than half hour in our computer and we do not have the output of $n_3$. We also test for the ideals contain an uncertain constant $a$. 
\begin{eqnarray*}
I^a_1&=&\left\langle x_2^2 x_3,x_1 x_4+x_2 a^2,x_1+x_3^2 a,x_1 x_2+x_2^2 x_3 x_4+x_2 x_4 a\right\rangle,\\
I^a_2&=&\left\langle x_2 x_3^2+x_3^3+x_3 a,x_1^2 a,x_1^3+x_3^2,x_2+x_3,x_1^2 x_3+x_1 x_3^2+a^3,x_1^2 x_2+x_1 x_3 a+x_1 a+x_2^3,x_2\right\rangle,\\
I^a_3&=&\left\langle x_1^2 x_2+x_1 x_3+x_1 a+x_2^2,x_1^2+x_1 a^2+x_1,x_1^2+a,x_3,x_1^2 x_2+x_3^2 a\right\rangle,\\
I^a_4&=&\left\langle x_1 x_2^2+x_1 x_2 a+x_1 x_3 a+x_3^2 a,x_1^2 x_3+x_1 x_3 a,x_2 x_3 a\right\rangle,\\
I^a_{5}&=&\left\langle x_1^2 x_3 a+x_1 x_2^2+x_2^3 x_3,x_1^3 x_2+x_2+x_3 a^2,x_1^3 x_3+x_1^3 a,x_1^2 x_3+x_1 x_2 x_3 a+x_3^3 a+x_3^2 a^2,x_1 x_3 a^2,x_1^2 x_2^2\right\rangle,\\
I^a_{6}&=&\left\langle x_1 x_2 x_3+x_3 a^2,x_1 x_2 x_3^2+x_2^2 a,x_2,x_2^2,x_3^2 a^2\right\rangle.
\end{eqnarray*}
$$
\begin{array}{|c|c|c|c|c|c|c|}\hline
 & T_1/s & n_1 & T_2/s & n_2&T_3/s& n_3\\ \hline
I^a_1& 0.32 & 6 & 0.14& 8 & 0.41& 6\\ \hline
 I^a_2& 0.09 & 7  &0.01& 9 & 0.38& 7\\ \hline
 I^a_3& 0.03 & 4 & 0.24 & 5 &0.40& 4\\ \hline
 I^a_4& 0.09 & 3  & 0.14 & 4 & 0.29& 3\\ \hline
 I^a_{5}& 0.11 & 8 &  0.24 & 9 &? & ?\\ \hline
 I^a_{6}& 0.02 & 4 & 0.14& 5 & 0.29& 4\\ \hline
\end{array},
$$
Our program was written in Mathematica. So there is still of much space to improve our program in $C^{++}$.

\subsection{Remarks On Syzygies Of Modules}
From the algorithm of syzygies for an ideal, it is straightforward to extend this algorithm to calculate the syzygies of any module. In fact, we calculate the syzygies of each column of the module. Then we combine them together to get the syzygies of the module. First, taking the the elements of the first column as an ideal, we calculate the syzygy directly by $\text{C2Z}$ . We denote this syzygy set by $\mathbf{B}^1$.  Then we perform dot product for all the syzygies in $\mathbf{B}^1$ with the second column. The obtained result generates a new ideal. The syzygy set of this new ideal is denoted as $\mathbf{C}^2$. Then, the syzygy module of the first two columns is $\mathbf{S}^e_j=\mathbf{C}^2_{j,i}\cdot \mathbf{S}^1_i$, where the $j$ denotes the syzygy index and $i$ is the component index for each syzygy. Repeating these operations to the last column, we obtain the final syzygies of the module. In each step, our algorithm can guarantee the minimal  generating set  of $\mathbf{C}^2$. There are extra  reductions after taking the dot product $\mathbf{C}^2_{j,i}\cdot \mathbf{S}^1_i$. An obvious example is a syzygy in $\mathbf{C}^2$ was just a syzygy of $\mathbf{S}^1$, dot product of them get a row of zero entry, which is obvious redundant. If two syzygies in $\mathbf{C}^2$ differ by a syzygy of $\mathbf{S}^1$, one of them is reduced by the other one in the final syzygy module. But in our method, this redundancy still remains for the syzygy of a module. To remove this redundancy, we need to get the syzygies $S_S$ of $\mathbf{S}^1_i$. In fact what we need is the $\mathbf{C}^2_{j,i} \mod S_S$. To get an efficient algorithm to eliminate the redundancy is beyond the scope of this paper. We leave this to future works. 

\section{Revisiting The IBP Relations}\label{sec2}
One important application of using syzygies is to obtain the  IBP relations in generalized unitarity cuts\cite{Gluza:2010ws,Schabinger:2011dz,Zhang:2014xwa,CaronHuot:2012ab}. The IBP relations are from the fact that  any total partial derivative of rational functions of loop momentum is vanishing under loop integration
\begin{align}
\int dl_1^d\cdots dl_L^d \partial_{l_I^\mu} v_J^\mu{C_{IJ}\over D_1 D_2\cdots D_m}=0~,
\end{align}
where $v^\mu_J\in\{l^\mu_1,l^\mu_2, \cdots, l^\mu_L, p^\mu_1,\cdots ,p^\mu_\beta\}$ for general $n$ point amplitude ($n$ point amplitudes have $\beta$ irreducible external-leg momentums, then we have $\beta \le 4$).  We define the $O_{IJ}\equiv \partial_{l_I^\mu} v_J^\mu$ as the generators of the IBP relations.  In Lorentz invariant parameters, $s_a\in\{l_{I_1}\cdot l_{I_2}, l_{I_1}\cdot p_J\}$ and $D_a$,  the IBP generators are defined in Baikov's method \cite{IBP1,IBP2,IBP3} as
\begin{align}
{{O}_{IJ}}=d {{\delta }_{IJ}}+\frac{\partial {{s}_{b}}}{\partial l_{I}^{\mu }}v_{J}^{\mu }\frac{\partial {{D}_{a}}}{\partial {{s}_{b}}} {\partial\over\partial D_a}~,
\end{align}
where $I\in\{1,2,\cdots,L\}$, $J\in\{1,2,\cdots,L+\beta\}$, $a,b\in\{1,2,\cdots, m\}$ and $d$ is the dimension of spacetime. When acting IBP generators on the dominators, it is useful to keep the dominators free of double poles according to the combination of the generators
\begin{align}
\sum_{IJ}c_{IJ}\frac{\partial {{s}_{b}}}{\partial l_{I}^{\mu }}v_{J}^{\mu }\frac{\partial {{D}_{a}}}{\partial {{s}_{b}}}+c_a D_a=0~.
\end{align}
We need to find all the solutions of these equations for $c_{IJ}, c_a$. With the definition
\begin{align}
{{Q}_{IJ}^a}=\frac{\partial {{s}_{b}}}{\partial l_{I}^{\mu }}v_{J}^{\mu }\frac{\partial {{D}_{a}}}{\partial {{s}_{b}}}~,
\end{align}
it is equivalent to find syzygies for the module
\begin{align}
M=\left(\begin{array}{cccc}
Q^1_{11}&Q^2_{11}&\cdots&Q^m_{11}\\
Q^1_{12}&Q^2_{12}&\cdots&Q^m_{12}\\
\cdots &\cdots&\cdots&\cdots \\
Q^1_{L,L+\beta}&Q^2_{L,L+\beta}&\cdots&Q^m_{L,L+\beta}\\
D_1&0&\cdots&0\\
0&D_2&\cdots&0\\
\cdots &\cdots&\cdots&\cdots \\
0&0&\cdots&D_m\\
\end{array}\right)~.
\end{align}

As a result, our analysis on syzygies comes to the story. In order to illustrate this application, we take a characteristic two-loop diagram as an example in Fig. \ref{feymann}. We will show that such diagram can reduce to zero under the generalized unitarity cut. Hence in physics, the four-point massless diagram is not a master integral. 
\begin{figure}
  \centering
  \includegraphics[width=7cm]{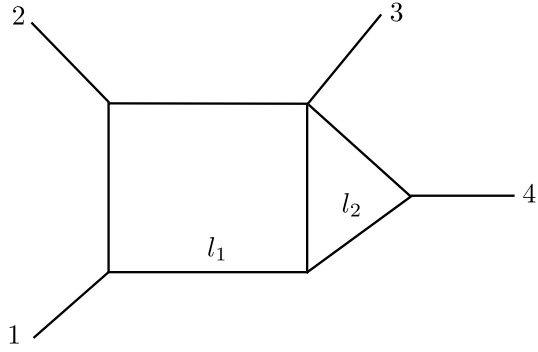}\\
  \caption{An example two-loop diagram.}  \label{feymann}
\end{figure}
The loop integration of this diagram is
\begin{align}
L=\int{{{d}^{4}}{{l}_{1}}}{{{d}^{4}}{{l}_{2}}}\frac{{{C }}}{D_1D_2D_3D_4D_5D_6}~,
\end{align}
where
\begin{align}
&{{D}_{1}}=l_{1}^{2}~,~~{{D}_{2}}={{({{l}_{1}}-{{p}_{1}})}^{2}}~,~~{{D}_{3}}={{({{l}_{1}}-{{p}_{1}}-{{p}_{2}})}^{2}}~,\nonumber\\
&{{D}_{4}}=l_{2}^{2}~,~~{{D}_{5}}={{({{l}_{2}}+{{p}_{4}})}^{2}}~,~~{{D}_{6}}={{({{l}_{1}}+{{l}_{2}})}^{2}}~,\nonumber\\
&{{D}_{7}}={{({{l}_{2}}+{{p}_{1}})}^{2}}~,~~{{D}_{8}}={{({{l}_{1}}+{{p}_{4}})}^{2}}~,~~{{D}_{9}}={{({{l}_{2}}+{{p}_{2}})}^{2}}~.
\end{align}
Extra $D_7$, $D_8$, $D_9$ are \emph{fake} propagators in order to keep the one-to-one correspondence between $D$ and $s$, where projections $s$ are defined as
\begin{align}
&{{s}_{1}}=l_{1}^{2}~,~~{{s}_{2}}={{l}_{1}}\cdot {{l}_{2}}~,~~{{s}_{3}}={{p}_{1}}\cdot {{l}_{1}}~,\nonumber\\
&{{s}_{4}}={{p}_{2}}\cdot {{l}_{1}}~,~~{{s}_{5}}={{p}_{4}}\cdot {{l}_{1}}~,~~{{s}_{6}}=l_{2}^{2}~,\nonumber\\
&{{s}_{7}}={{p}_{1}}\cdot {{l}_{2}}~,~~{{s}_{8}}={{p}_{2}}\cdot {{l}_{2}}~,~~{{s}_{9}}={{p}_{4}}\cdot {{l}_{2}}~.
\end{align}
$D_7$ is used to keep the one-to-one relation between $D_a$ and $s_a$. If we use symbols ${{v}_{J}}=\{{{l}_{1}},{{l}_{2}},{{p}_{1}},{{p}_{2}}\}$, where $J=1,2,3,4$, and ${{l}_{I}}=\{{{l}_{1}},{{l}_{2}}\}$, where $I=1,2$, then the effects of IBP operators $O_{IJ}$ on the propagators $D_a$ are given by
\begin{align}
{{O}_{IJ}}{{D}_{a}}=d{{\delta }_{IJ}}{{D}_{a}}+v_{J}^{\mu }\frac{\partial {{s}_{b}}}{\partial l_{I}^{\mu }}\frac{\partial {{D}_{a}}}{\partial {{s}_{b}}}~,
\end{align}
where $d$ is the dimension of spacetime. As a result, the module $M$ is 
\begin{align}
\left(
\begin{array}{ccl}
Q_{11}&\rightarrow&\{2D_{1} , D_{1}{\tiny\text{+}}D_{2} , D_{1}{\tiny\text{+}}D_{3}\text{-}2 k_{12} , 0 , 0 , D_{1}\text{-}D_{4}{\tiny\text{+}}D_{6} , 0 , D_{1}{\tiny\text{+}}D_{8} , 0\}\\
Q_{12}&\rightarrow&\{\text{-}D_{1}\text{-}D_{4}{\tiny\text{+}}D_{6} , \text{-}D_{1}{\tiny\text{+}}D_{6}\text{-}D_{7} , \text{-}D_{1}{\tiny\text{+}}D_{4}{\tiny\text{+}}D_{6}\text{-}D_{7}\text{-} D_{9} , 0 , 0 , \text{-}D_{1}{\tiny\text{+}}D_{4}{\tiny\text{+}}D_{6} , 0 , \text{-}D_{1}\text{-}2 D_{4}{\tiny\text{+}}D_{5}{\tiny\text{+}}D_{6} , 0 \}\\
Q_{13}&\rightarrow&\{D_{1}\text{-}D_{2} , D_{1}\text{-}D_{2} , D_{1}\text{-}D_{2}\text{-}2 k_{12} , 0 , 0 , D_{1}\text{-}D_{2}\text{-}D_{4}{\tiny\text{+}}D_{7} , 0 , D_{1}\text{-}D_{2}{\tiny\text{+}}2 k_{14} , 0 \}\\
Q_{14}&\rightarrow&\{ D_{2}\text{-}D_{3}{\tiny\text{+}}2 k_{12} , D_{2}\text{-}D_{3} , D_{2}\text{-}D_{3} , 0 , 0 , D_{2}\text{-}D_{3}\text{-}D_{4}{\tiny\text{+}}D_{9}{\tiny\text{+}}2 k_{12} , 0 , D_{2}\text{-}D_{3}\text{-}2 k_{14} , 0 \}\\
Q_{15}&\rightarrow&\{ \text{-}D_{1}{\tiny\text{+}}D_{8} , \text{-}D_{1}{\tiny\text{+}}D_{8}\text{-}2 k_{14} , \text{-}D_{1}{\tiny\text{+}}D_{8}{\tiny\text{+}}2 k_{12} , 0 , 0 , \text{-}D_{1}\text{-}D_{4}{\tiny\text{+}}D_{5}{\tiny\text{+}}D_{8} , 0 , \text{-}D_{1}{\tiny\text{+}}D_{8} , 0 \}\\
Q_{22}&\rightarrow& \{0 , 0 , 0 , 2 D_{4} , D_{4}{\tiny\text{+}}D_{5} , \text{-}D_{1}{\tiny\text{+}}D_{4}{\tiny\text{+}}D_{6} , D_{4}{\tiny\text{+}}D_{7} , 0 , D_{4}{\tiny\text{+}}D_{9}\} \\
Q_{23}&\rightarrow& \{0 , 0 , 0 , \text{-}D_{4}{\tiny\text{+}}D_{7} , \text{-}D_{4}{\tiny\text{+}}D_{7}{\tiny\text{+}}2 k_{14} , D_{1}\text{-}D_{2}\text{-}D_{4}{\tiny\text{+}}D_{7} , \text{-}D_{4}{\tiny\text{+}}D_{7} , 0 , \text{-}D_{4}{\tiny\text{+}}D_{7}{\tiny\text{+}}2 k_{12}\} \\
Q_{24}&\rightarrow&\{0 , 0 , 0 , \text{-}D_{4}{\tiny\text{+}}D_{9} , \text{-}D_{4}{\tiny\text{+}}D_{9}\text{-}2 (k_{12}{\tiny\text{+}}k_{14}) , D_{2}\text{-}D_{3}\text{-}D_{4}{\tiny\text{+}}D_{9}{\tiny\text{+}}2 k_{12} , \text{-}D_{4}{\tiny\text{+}}D_{9}{\tiny\text{+}}2 k_{12} , 0 , \text{-}D_{4}{\tiny\text{+}}D_{9} \}\\
 Q_{25}&\rightarrow&\{0 , 0 , 0 , \text{-}D_{4}{\tiny\text{+}}D_{5} , \text{-}D_{4}{\tiny\text{+}}D_{5} , \text{-}D_{1}\text{-}D_{4}{\tiny\text{+}}D_{5}{\tiny\text{+}}D_{8} , \text{-}D_{4}{\tiny\text{+}}D_{5}{\tiny\text{+}}2 k_{14} , 0 , \text{-}D_{4}{\tiny\text{+}}D_{5}\text{-}2 (k_{12}{\tiny\text{+}}k_{14})\}\\
 Q_1&\rightarrow&\{D_1 , 0 , 0 , 0, 0 , 0, 0 , 0 , 0\}\\
&\cdots&\\
      Q_6&\rightarrow&\{0 , 0 , 0 , 0, 0 , D_6, 0 , 0 , 0\}\nonumber
\end{array}
\right)~,
\end{align}
where $k_{12}=p_1\cdot p_2$ and $k_{14}=p_1\cdot p_4$.

Finally, we found 64 syzygies. We have verified that the syzygy module is same as the output of \textit{{\sc Singular}} with 87 syzygies. We first use one of them which is of first order
\begin{align}
\textbf{C}^{(1)}&=\left(0, 0, 0, 0, 0, \text{-}D_4 {\tiny{\text{+}}} D_5, \text{-}D_1 \text{-} 2 D_4 {\tiny{\text{+}} } 2 D_5 {\tiny{\text{+}}} D_8, 0, 0, D_1 \text{-} D_4 \text{-} D_6, 0\right), \nb\\
\textbf{c}^{(1)}&=\left(0, 0, 2 D_1 {\tiny{\text{+}} } 2 D_4 \text{-} 2 D_5 \text{-} 2 D_8, 2 D_1 {\tiny{\text{+}} }2 D_4 \text{-} 2 D_5 \text{-} 2 D_8, 2 D_4 \text{-} 2 D_5\right).\nb
\end{align}

According to the integrand reduction \cite{Mastrolia:2012bu,Ellis:2008ir,Giele:2008ve,Bevilacqua:2011xh,Hirschi:2011pa,Ossola:2006us,Badger:2010nx,Mastrolia:2010nb,Ellis:2007br,Ossola:2007ax,Zhang:2012ce}, the naive power counting of renormalizable condition constrains the irreducible integrals as following, $\frac{D_7^{n_7} D_8^{n_8} D_9^{n_9}}{D_1 D_2\cdots D_6}$ with constraints $n_7\leq 4$, $n_8+n_9\leq 2$, and $n_7+n_8+n_9\leq 4$. An IBP relation induce a integral equation among the irreducible integrand
\begin{align}
0=\int \sum_{IJ}\frac{O_{IJ}c_{IJ} }{D_1\cdots D_6}+\sum_{a=1}^6{c_a\over D_1\cdots D_6},
\end{align}
where $c_{IJ}$ and $c_a$ are the components of $\textbf{C}^{(1)}$ and $\textbf{c}^{(1)}$ respectively.

We shall show how IBP relations work by demonstrating the IBP relation generated by the first syzygy which we pick. First, we calculate the action of operator $O_{IJ}$ on vector $C_{IJ}^{(1)}$ and the result is simple: $-d D_1-2 d D_4+2 d D_5+d D_8$. Then we add it with $\sum _{i=1}^6 c_i$, and we have an IBP relation $$\int \frac{(4-d) D_1+(6-2 d) D_4+(2 d-6) D_5+(d-4) D_8}{D_1 D_2 \ldots D_6  }=0.$$ Similarly, IBP relations in diagrams with less propagators can reduce $\int\frac{D5}{D_1 D_2 \ldots D_6  }$ to $$-\frac {(3 d - 10) (3 d - 8) } {4 (d - 4)^2 k_{12}^2}\int\frac{1}{D_3 D_4 D_6}.$$ Finally, the IBP relation is: $$\int\frac{(d-4) D_8}{D_1 D_2 \ldots D_6  }= \int\frac{(4-d)} {D_2 D_3 D_4 D_5 D_6  }+\frac{(6-2 d)}{D_1 D_2 D_3 D_5 D_6}+\frac {(2 d-6)(3 d - 10) (3 d - 8) } {4 (d - 4)^2 k_{12}^2D_3 D_4 D_6}.$$This result is consistent with the output of FIRE5 \cite{Smirnov:2014hma}. Our program spent 42 seconds while FIRE5 spent 36 seconds in getting this result.

To reduce the integral further, we need to use more IBP relations. For the general $r$-order  syzygies, we use $m_{i_1\cdots i_r}$  denotes the product of $\{D_1, D_2,\cdots, D_9, k_{12}, k_{14}\}$, where the subindex $i\in \{1\cdots 9\}$ represent a $D_i$ factor in the product, $i\in {a, b}$ denotes a factor  $k_{12}$ and $k_{14}$ in the product respectively. For example $m_{1a}=D_1 k_{12}, m_{12}=D_1 D_2$.We choose the following syzygies
 \begin{align}
\textbf{C}^{(2)}=&(-m_{1b}+m_{2a}+m_{3b}-m_{8a},0,m_{1b}+m_{8a},m_{1b},m_{2a},0,\nb\\&-2 m_{4a}-2 m_{4b}+m_{5a}+m_{7a}+m_{7b}+m_{9b}+2 m_{ab},-2 m_{4a}-m_{4b}+m_{5a},-m_{4b},-m_{7a})\nb\\
\textbf{c}^{(2)}=&(m_{1b}-m_{2a}-m_{3b}+m_{8a}-2 m_{ab},m_{1b}-m_{2a}-m_{3b}+m_{8a}+2 m_{ab},\nb\\ &m_{1b}-m_{2a}-m_{3b}+m_{8a}+2 m_{ab}, 2 m_{4a}+2 m_{4b}-m_{5a}-m_{7a}-m_{7b}-m_{9b}-4 m_{ab},\nb\\ &2 m_{4a}+2 m_{4b}-m_{5a}-m_{7a}-m_{7b}-m_{9b}-4 m_{ab},\nb\\ &m_{1b}-m_{2a}-m_{3b}+2 m_{4a}+2 m_{4b}-m_{5a}-m_{7a}-m_{7b}+m_{8a}-m_{9b}-2 m_{ab}),\nb
\end{align}
 \begin{align}
\textbf{C}^{(3)}=&(-2 m_{14}-m_{17}+4 m_{24}-m_{27}-m_{29}-2 m_{34}+m_{37}+4 m_{4a}+m_{78}+m_{89},0,\nb\\ &2 m_{14}+m_{17}-2 m_{24}+2 m_{34}-4 m_{4a}-m_{78}-m_{89},m_{17}-2 m_{24},-m_{27}-m_{29},0,\nb\\ &2 m_{47}-m_{57}-m_{59}+2 m_{7a},m_{47}-m_{57}-m_{59},m_{47},m_{77}+m_{79})\nb\\
\textbf{c}^{(3)}=&(2 m_{14}+m_{17}-4 m_{24}+m_{27}+m_{29}+2 m_{34}-m_{37}-4 m_{4a}-m_{78}-2 m_{7a}-m_{89},\nb\\ &2 m_{14}+m_{17}-4 m_{24}+m_{27}+m_{29}+2 m_{34}-m_{37}-8 m_{4a}-m_{78}-2 m_{7b}-m_{89}-2 m_{9b},\nb\\ &2 m_{14}+m_{17}-4 m_{24}+m_{27}+m_{29}+2 m_{34}-m_{37}-4 m_{4a}-m_{78}+2 m_{7a}-m_{89},\nb\\ &-2 m_{47}+m_{57}+m_{59}-4 m_{7a},-2 m_{47}+m_{57}+m_{59}-2 m_{7a}+2 m_{7b}+2 m_{9b},\nb\\ &2 m_{14}+m_{17}-4 m_{24}+m_{27}+m_{29}+2 m_{34}-m_{37}-2 m_{47}-4 m_{4a}+m_{57}+m_{59}\nb\\ &-m_{78}-2 m_{7a}-m_{89}).\nb
\end{align}

 \begin{align}
\textbf{C}^{(4)}=&(-4 m_{119}-2 m_{124}-2 m_{127}+4 m_{129}+4 m_{134}+2 m_{137}-4 m_{139}+2 m_{145}+2 m_{148}\nb\\ &-8 m_{14a}-4 m_{17a}+2 m_{189}+8 m_{19a}+8 m_{224}-2 m_{227}-2 m_{229}-12 m_{234}+4 m_{237}\nb\\ &+2 m_{239}-3 m_{245}-3 m_{248}+20 m_{24a}+2 m_{278}-4 m_{27a}+2 m_{289}+4 m_{334}-2 m_{337}\nb\\ &+2 m_{345}-16 m_{34a}-2 m_{378}+4 m_{37a}-m_{458}-4 m_{45a}-m_{488}+16 m_{4aa}+4 m_{78a},\nb\\ &0,4 m_{119}+2 m_{124}+2 m_{127}-2 m_{129}-4 m_{134}-2 m_{137}+4 m_{139}-2 m_{145}-2 m_{148}\nb\\ &+8 m_{14a}+4 m_{17a}-2 m_{189}-8 m_{19a}-4 m_{224}+8 m_{234}+2 m_{245}+2 m_{248}-16 m_{24a}\nb\\ &-2 m_{278}-2 m_{289}-4 m_{334}-2 m_{345}+16 m_{34a}+2 m_{378}+m_{458}+4 m_{45a}+m_{488}\nb\\ &-16 m_{4aa}-4 m_{78a},-2 m_{124}+2 m_{127}-2 m_{129}-2 m_{137}+4 m_{17a}-4 m_{224}+4 m_{234}\nb\\ &+2 m_{245}+2 m_{248}-8 m_{24a},-2 m_{129}-2 m_{227}-2 m_{229}+2 m_{234}+2 m_{237}+m_{245}\nb\\  &+m_{248}-4 m_{24a}-4 m_{27a},0,-2 m_{159}+2 m_{179}-2 m_{244}+4 m_{247}+4 m_{249}-4 m_{24a}\nb\\ &-2 m_{257}-2 m_{259}-2 m_{279}+4 m_{27a}-2 m_{299}+4 m_{29a}+2 m_{345}-4 m_{347}+2 m_{357}\nb\\ &+2 m_{379}-4 m_{37a}+m_{455}-m_{457}+m_{458}-4 m_{45a}-m_{478}+8 m_{47a}-4 m_{57a}-4 m_{79a}+8 m_{7aa},\nb\\ &-2 m_{159}-2 m_{244}+2 m_{247}+2 m_{249}-2 m_{257}-2 m_{259}+2 m_{345}-2 m_{347}+2 m_{357}\nb\\ &+m_{455}+m_{458}-4 m_{45a}+4 m_{47a}-4 m_{57a},-2 m_{244}+2 m_{247}+2 m_{249}-2 m_{347}+4 m_{47a},\nb\\ &2 m_{179}+2 m_{244}-2 m_{247}-2 m_{249}+2 m_{277}+2 m_{279}-2 m_{377}-m_{457}-m_{478}+4 m_{77a}),\nb\\
\textbf{c}^{(4)}=&(4 m_{119}+2 m_{124}+2 m_{127}-4 m_{129}-4 m_{134}-2 m_{137}+4 m_{139}-2 m_{145}-2 m_{148}\nb\\ &+8 m_{14a}+4 m_{17a}-2 m_{189}-8 m_{19a}-8 m_{224}+2 m_{227}+2 m_{229}+12 m_{234}-4 m_{237}\nb\\ &-2 m_{239}+3 m_{245}+3 m_{248}-16 m_{24a}-2 m_{278}-2 m_{289}-4 m_{29a}-4 m_{334}+2 m_{337}\nb\\ &-2 m_{345}+16 m_{34a}+2 m_{378}+m_{458}+4 m_{45a}+m_{488}-16 m_{4aa}-4 m_{78a}-8 m_{7aa},\nb\\ &4 m_{119}+2 m_{124}+2 m_{127}-4 m_{129}-4 m_{134}-2 m_{137}+4 m_{139}-2 m_{145}-2 m_{148}\nb\\ &+8 m_{14a}+4 m_{17a}-2 m_{189}-16 m_{19a}-4 m_{19b}-8 m_{224}+2 m_{227}+2 m_{229}+12 m_{234}\nb\\ &-4 m_{237}-2 m_{239}+3 m_{245}+3 m_{248}-28 m_{24a}-2 m_{278}+4 m_{27a}-4 m_{27b}-2 m_{289}\nb\\ &-4 m_{29b}-4 m_{334}+2 m_{337}-2 m_{345}+24 m_{34a}+4 m_{34b}+2 m_{378}-4 m_{37a}+4 m_{37b}\nb\\ &+m_{458}+8 m_{45a}+2 m_{45b}+m_{488}+4 m_{48a}+2 m_{48b}-32 m_{4aa}-8 m_{4ab}-4 m_{78a}-8 m_{7ab},\nb\\ &4 m_{119}+2 m_{124}+2 m_{127}-4 m_{129}-4 m_{134}-2 m_{137}+4 m_{139}-2 m_{145}-2 m_{148}\nb\\ &+8 m_{14a}+4 m_{17a}-2 m_{189}-8 m_{19a}-8 m_{224}+2 m_{227}+2 m_{229}+12 m_{234}-4 m_{237}\nb\\ &-2 m_{239}+3 m_{245}+3 m_{248}-24 m_{24a}-2 m_{278}+8 m_{27a}-2 m_{289}+4 m_{29a}-4 m_{334}+2 m_{337}\nb\\ &-2 m_{345}+16 m_{34a}+2 m_{378}-8 m_{37a}+m_{458}+4 m_{45a}+m_{488}-16 m_{4aa}-4 m_{78a}+8 m_{7aa},\nb\\ &2 m_{159}-2 m_{179}+2 m_{244}-2 m_{245}-4 m_{247}-4 m_{249}+8 m_{24a}+4 m_{257}+4 m_{259}+2 m_{279}\nb\\ &-8 m_{27a}+2 m_{299}-8 m_{29a}-2 m_{345}+4 m_{347}-4 m_{357}-2 m_{379}+8 m_{37a}-m_{455}+m_{457}\nb\\ &-m_{458}+4 m_{45a}+m_{478}-8 m_{47a}+8 m_{57a}+4 m_{79a}-16 m_{7aa},2 m_{159}-2 m_{179}+4 m_{19b}\nb\\ &-2 m_{247}-2 m_{249}+4 m_{24a}+2 m_{257}+2 m_{259}+2 m_{279}-4 m_{27a}+4 m_{27b}+2 m_{299}-4 m_{29a}\nb\\ &+4 m_{29b}-2 m_{345}+2 m_{347}-4 m_{34b}-2 m_{357}-2 m_{379}+4 m_{37a}-4 m_{37b}-m_{455}+m_{457}-m_{458}\nb\\ &+4 m_{45a}-2 m_{45b}+m_{478}-4 m_{47a}-2 m_{48b}+8 m_{4ab}+4 m_{57a}+4 m_{79a}-8 m_{7aa}+8 m_{7ab},\nb\\ &4 m_{119}+2 m_{124}+2 m_{127}-4 m_{129}-4 m_{134}-2 m_{137}+4 m_{139}-2 m_{145}-2 m_{148}+8 m_{14a}+2 m_{159}\nb\\ &-2 m_{179}+4 m_{17a}-2 m_{189}-8 m_{19a}-8 m_{224}+2 m_{227}+2 m_{229}+12 m_{234}-4 m_{237}-2 m_{239}\nb\\ &+2 m_{244}+3 m_{245}-4 m_{247}+3 m_{248}-4 m_{249}-16 m_{24a}+2 m_{257}+2 m_{259}-2 m_{278}+2 m_{279}\nb\\ &-2 m_{289}+2 m_{299}-4 m_{29a}-4 m_{334}+2 m_{337}-4 m_{345}+4 m_{347}+16 m_{34a}-2 m_{357}+2 m_{378}\nb\\ &-2 m_{379}-m_{455}+m_{457}+8 m_{45a}+m_{478}-8 m_{47a}+m_{488}-16 m_{4aa}+4 m_{57a}-4 m_{78a}\nb\\ &+4 m_{79a}-8 m_{7aa}).\nb
\end{align}

 \begin{align}
\textbf{C}^{(5)}=&(22 m_{11}-34 m_{12}+14 m_{13}+32 m_{14}-6 m_{15}+46 m_{17}-11 m_{18}-20 m_{19}-76 m_{1a}+8 m_{1b}\nb\\ &-84 m_{24}+24 m_{25}+6 m_{27}+16 m_{28}+6 m_{29}-50 m_{2a}+72 m_{34}-18 m_{35}-26 m_{37}\nb\\ &-7 m_{38}-88 m_{3b}-20 m_{48}-104 m_{4a}+24 m_{5a}+40 m_{6a}+80 m_{6b}-26 m_{78}-40 m_{7a}\nb\\ &-40 m_{7b}+14 m_{89}+66 m_{8a}-40 m_{9b},-40 m_{1b}-40 m_{3b}-40 m_{8a}+80 m_{ab},\nb\\ & -22 m_{11}+17 m_{12}-18 m_{13}-12 m_{14}+6 m_{15}-46 m_{17}+11 m_{18}+36 m_{1a}-48 m_{1b}+52 m_{24}\nb\\ &-12 m_{25}-8 m_{28}-52 m_{34}+12 m_{35}+9 m_{38}+104 m_{4a}+40 m_{4b}-24 m_{5a}-40 m_{6b}+26 m_{78}\nb\\ &+6 m_{89}-66 m_{8a},-4 m_{11}+17 m_{12}+40 m_{14}-6 m_{15}-46 m_{17}+2 m_{18}-48 m_{1b}+52 m_{24}\nb\\ &-12 m_{25}-8 m_{28}-20 m_{48}+40 m_{4b}-40 m_{6b}+20 m_{78},m_{12}+20 m_{14}-20 m_{19}-40 m_{1a}\nb\\ &+m_{23}+20 m_{24}+6 m_{27}+6 m_{29}-50 m_{2a}-20 m_{34}+20 m_{37}+40 m_{6a}-40 m_{7a}, \nb\\ &16 m_{7a}+16 m_{7b}+16 m_{9b},4 m_{14}+5 m_{15}-17 m_{17}+8 m_{19}+8 m_{1a}+16 m_{1b}-18 m_{24}\nb\\ &+9 m_{27}+9 m_{29}-18 m_{2a}+12 m_{34}-3 m_{35}-9 m_{37}-16 m_{3b}-80 m_{44}+32 m_{45}+8 m_{47}+4 m_{48}\nb\\ &+80 m_{49}+136 m_{4a}+64 m_{4b}+32 m_{57}-32 m_{59}-94 m_{5a}-20 m_{77}+6 m_{78}-20 m_{79}-90 m_{7a}\nb\\ &-56 m_{7b}-10 m_{89}+4 m_{8a}-56 m_{9b}-144 m_{ab},\nb\\ &4 m_{14}+5 m_{15}+8 m_{19}+16 m_{1a}+16 m_{1b}-9 m_{24}+8 m_{34}+m_{35}-40 m_{44}+6 m_{45}\nb\\ &-8 m_{46}-6 m_{47}+2 m_{48}+40 m_{49}+64 m_{4a}-8 m_{4b}+8 m_{56}+26 m_{57}-6 m_{59}-50 m_{5a}-16 m_{6a}\nb\\ &-16 m_{6b}-8 m_{89},-4 m_{14}+4 m_{15}-8 m_{17}+16 m_{1b}-9 m_{24}-14 m_{45}+8 m_{46}-46 m_{47}+2 m_{48}\nb\\ &-8 m_{4b}-8 m_{56}+32 m_{57}-16 m_{6b}+8 m_{78},\nb\\ &4 m_{14}-5 m_{17}-4 m_{19}+4 m_{34}-m_{37}-40 m_{44}+20 m_{47}+40 m_{49}+72 m_{4a}-8 m_{67}+8 m_{69}\nb\\ &-26 m_{77}-26 m_{79}+50 m_{7a}),\nb\\
\textbf{c}^{(5)}=&(-22 m_{11}+34 m_{12}-14 m_{13}-32 m_{14}+6 m_{15}-46 m_{17}+11 m_{18}+20 m_{19}+84 m_{1a}-8 m_{1b}\nb\\ &+84 m_{24}-24 m_{25}-6 m_{27}-16 m_{28}-6 m_{29}+52 m_{2a}-72 m_{34}+18 m_{35}+26 m_{37}+7 m_{38}\nb\\ &+88 m_{3b}+20 m_{48}+24 m_{4a}-80 m_{4b}-12 m_{5a}-40 m_{6a}-80 m_{6b}+26 m_{78}+132 m_{7a}+80 m_{7b}\nb\\ &-14 m_{89}-70 m_{8a}+80 m_{9b}+176 m_{ab},-22 m_{11}+34 m_{12}-14 m_{13}-32 m_{14}+6 m_{15}-46 m_{17}\nb\\ &+11 m_{18}+20 m_{19}+112 m_{1a}-6 m_{1b}+84 m_{24}-24 m_{25}-6 m_{27}-16 m_{28}-6 m_{29}+50 m_{2a}\nb\\ &-72 m_{34}+18 m_{35}+26 m_{37}+7 m_{38}+90 m_{3b}+20 m_{48}+208 m_{4a}+40 m_{4b}-48 m_{5a}\nb\\ &-40 m_{6a}-80 m_{6b}+26 m_{78}+40 m_{7a}+52 m_{7b}-14 m_{89}-82 m_{8a}+52 m_{9b}-100 m_{ab},\nb\\ &-22 m_{11}+34 m_{12}-14 m_{13}-32 m_{14}+6 m_{15}-46 m_{17}+11 m_{18}+20 m_{19}+68 m_{1a}-8 m_{1b}\nb\\ &+84 m_{24}-24 m_{25}-6 m_{27}-16 m_{28}-6 m_{29}+48 m_{2a}-72 m_{34}+18 m_{35}+26 m_{37}+7 m_{38}\nb\\ &+88 m_{3b}+20 m_{48}+184 m_{4a}+80 m_{4b}-36 m_{5a}-40 m_{6a}-80 m_{6b}+26 m_{78}-52 m_{7a}-14 m_{89}\nb\\ &-62 m_{8a}-176 m_{ab},-4 m_{14}-5 m_{15}+17 m_{17}-8 m_{19}+18 m_{24}-9 m_{27}-9 m_{29}+36 m_{2a}\nb\\ &-12 m_{34}+3 m_{35}+9 m_{37}+32 m_{3b}+80 m_{44}-32 m_{45}-8 m_{47}-4 m_{48}-80 m_{49}-136 m_{4a}\nb\\ &-144 m_{4b}-32 m_{57}+32 m_{59}+66 m_{5a}-16 m_{6a}-32 m_{6b}+20 m_{77}-6 m_{78}+20 m_{79}+182 m_{7a}\nb\\ &+136 m_{7b}+10 m_{89}-8 m_{8a}+136 m_{9b}+288 m_{ab},-4 m_{14}-5 m_{15}+17 m_{17}-8 m_{19}-18 m_{1b}\nb\\ &+18 m_{24}-9 m_{27}-9 m_{29}+18 m_{2a}-12 m_{34}+3 m_{35}+9 m_{37}+14 m_{3b}+80 m_{44}-32 m_{45}-8 m_{47}\nb\\ &-4 m_{48}-80 m_{49}-192 m_{4a}-104 m_{4b}-32 m_{57}+32 m_{59}+94 m_{5a}-16 m_{6a}-32 m_{6b}+20 m_{77}\nb\\ &-6 m_{78}+20 m_{79}+154 m_{7a}+68 m_{7b}+10 m_{89}-4 m_{8a}+68 m_{9b}+244 m_{ab},-22 m_{11}+34 m_{12}\nb\\ &-14 m_{13}-28 m_{14}-7 m_{15}-37 m_{17}+11 m_{18}+20 m_{19}+84 m_{1a}-8 m_{1b}+86 m_{24}-8 m_{25}\nb\\ &-15 m_{27}-16 m_{28}-15 m_{29}+52 m_{2a}-76 m_{34}+13 m_{35}+35 m_{37}+7 m_{38}+88 m_{3b}+80 m_{44}\nb\\ &-32 m_{45}-8 m_{47}+16 m_{48}-80 m_{49}+16 m_{4a}-96 m_{4b}-32 m_{57}+32 m_{59}+46 m_{5a}-40 m_{6a}\nb\\ &-80 m_{6b}+20 m_{77}+28 m_{78}+20 m_{79}+130 m_{7a}+136 m_{7b}-12 m_{89}-70 m_{8a}+136 m_{9b}+176 m_{ab})
\end{align}
For convenience, we let $k_{12}=-1/2, k_{14}=-1/2, k_{24}=1$.
According to these chosen syzygies, we can obtain the IBP relations,
\begin{align*}
&\int\frac{4 D_7 (-4 + d)}{D_1 D_2 \ldots D_6  }=\int\frac{1}{D_1 D_2 \ldots D_6  }\times\\
&\left(2 (d-3) m_{12}+(7 d-36) m_{14}+\frac{1}{2} (d-6) m_{15} \right.+(d-3) m_{16}+(d-7) m_{17}-(d-2) m_{19}-(d-3) m_{23} \\
&+(65-16 d) m_{24}+2 (d-3) m_{26} +4 (d-4) m_{27} +4 (d-4) m_{29}+(8 d-30) m_{34}+\frac{1}{2} (d+2) m_{35} \\
&+(d-3) m_{36}+(19-5 d) m_{37}-(d-3) m_{38}+9 (d-2) m_{45} -14 (d-2) m_{47}-4 (d-2) m_{49}\\
&+\frac{7}{2} (d-2) m_{57} +\frac{3}{2} (d-2) m_{59}-4 (d-3) m_{68}+2 (d-2) m_{69} -(d-7) m_{78}-2 (d-5) m_{89}-(d-5) m_1\\
&+(8 d-31) m_4 \left.+\frac{1}{2} (9 d-35) m_5+(d-5) m_6-m_{25}+5 m_{48} \right)\\
&\int\frac{D_9 (-3 + d)}{D_1 D_2 \ldots D_6  }=\int\frac{1}{D_1 D_2 \ldots D_6  }\times\\
&\left(-(d-4) m_1+(d-4) m_2+13 (d-4) m_4+2 (2 d-9) m_5+\frac{1}{2} (7 d-33) m_7-2 (d-4) m_8+2 (d-3) m_{11}\right.\\
&-3 (d-3) m_{12}+2 (d-3) m_{13}+4 (3 d-13) m_{14}+\left(d-\frac{13}{2}\right) m_{15}+\frac{1}{2} (7 d-33) m_{17}-(d-3) m_{18}\\
&-(d-3) m_{19}+(98-25 d) m_{24}-\frac{1}{2} (d-4) m_5 m_{24}+4 (d-3) m_{26}+\frac{1}{2} (13 d-51) m_{27}+2 (d-3) m_{28}\\
&+\frac{1}{2} (13 d-51) m_{29}+2 (6 d-23) m_{34}-\frac{1}{2} (13 d-51) m_{37}-2 (d-3) m_{38}+(d-1) m_5 m_{44}\\
&+\left(9 d-\frac{35}{2}\right) m_{45}-\frac{1}{2} (d-1) m_7 m_{45}+\frac{1}{2} (d-4) m_8 m_{45}-14 (d-2) m_{47}\\
&+d m_{48}-4 (d-2) m_{49}-\frac{1}{2} (d-1) m_4 m_{55}+\frac{3}{2} (d-2) m_{57}+\frac{7}{2} (d-2) m_{59}\\
&\left.-4 (d-3) m_{68}+2 (d-2) m_{69}-\frac{3}{2} (3 d-13) m_{78}-\frac{1}{2} (11 d-45) m_{89}+m_{25}+\frac{5 m_{35}}{2}\right)\\
& \int\frac{1/2 (4 - d)}{D_1 D_2 \ldots D_6  }=\int\frac{1}{D_1 D_2 \ldots D_6  }\times\\
&\left(-\frac{1}{2} (2 d+1) m_1+\frac{1}{2} (d-4) m_2+\frac{3}{2} (d-4) m_3+\left(12-\frac{5 d}{2}\right) m_4 \right.\\
&+\left(\frac{29}{2}-4 d\right) m_5+\left(\frac{41}{2}-6 d\right) m_7-\frac{5}{4} (d-4) m_8-\frac{3}{2} (d-3) m_9\\
&+\frac{3}{2} (d-3) m_{11}-3 (d-3) m_{12}+\frac{1}{2} (d-3) m_{13}-\frac{3}{2} (3 d-17) m_{14}\\
&+\frac{1}{2} (d-1) m_{15}+\left(7-\frac{3 d}{2}\right) m_{17}-\frac{3}{4} (d-3) m_{18}+3 (d-3) m_{19}+\left(9 d-\frac{75}{2}\right) m_{24}\\
&+\left(\frac{17}{2}-2 d\right) m_{27}+2 (d-3) m_{28}+\left(\frac{17}{2}-2 d\right) m_{29}-\frac{1}{2} (7 d-27) m_{34}\\
&-\frac{3}{2} (d-2) m_{35}+\left(2 d-\frac{17}{2}\right) m_{37}-\frac{1}{4} (d-3) m_{38}-7 (d-2) m_{45}+10 (d-2) m_{47}\\
&+\frac{1}{2} (d-12) m_{48}+4 (d-2) m_{49}-2 (d-2) m_{57}-(d-2) m_{59}-2 (d-2) m_{69}+\frac{1}{4} (3 d-19) m_{78}\\
&\left.-\frac{1}{4} (3 d+1) m_{89}+\frac{m_6}{2}+\frac{m_{25}}{2}\right).
\end{align*}
According to former IBP relations and IBP relations of diagrams with less propagators, the result can be further reduced as following:
\begin{align*}
\int\frac{D_7}{D_1 D_2 \ldots D_6  }&=\frac{(d-4 )}{(d-3 ) } I_5^b-2I_5^a+ \frac{6d-20 }{d-4} I_4^a - \frac{(9d-30)(3 d-8 ) }{2 (d-4 )^2 }I_3^b + \frac{(6d-20 ) (3 d-8)}{(d-4)^2}I_3^a\\
\int\frac{D_9}{D_1 D_2 \ldots D_6  }&=\frac{3 (3 d-10) (3 d-8)}{(d-4)^2 }I_3^b-\frac{5 (3 d-10) (3 d-8) }{(d-4)^2}I_3^a+\frac{2 (3 d-10) (3 d-8) }{(d-4)^2}I_3^c\\
&-\frac{6 (d-4)}{(d-3)}I_5^b-\frac{2 (3 d-10) }{(d-4)}I_4^a+4I_5^a\\
\int\frac{1}{D_1 D_2 \ldots D_6  }&=-\frac{6 (d-3)}{(d-4) }I_5^a-6I_5^b+\frac{(9 d-30) (d-3)}{(d-4)^2 }I_4^a\\
&-\frac{3 (3 d-10) (3 d-8) (d-3)}{(d-4)^3 }I_3^a+\frac{3 (3 d-10) (3 d-8) (d-3)}{(d-4)^3 }I_3^c,
\end{align*}
where 
\begin{eqnarray*}
I_5^a&=&\int {1\over D_1 D_2 D_3 D_5 D_6},  I_5^b=\int {1\over D_2 D_3 D_4 D_5 D_6},
I_4^a=\int {1\over D_1 D_3 D_5 D_6}, \\
 I_3^a&=&\int {1\over D_3 D_4 D_6},~~~~~~ ~~I_3^b=\int {1\over D_3 D_5 D_6}, ~~~~~~~I_3^c=\int {1\over D_2 D_5 D_6}.
\end{eqnarray*}
This results are same as the output from FIRE5. Other higher order integrands can also be reduced under these IBP relations. Our program spent 125 seconds to reduce each integral into the irreducible basis while FIRE5 spent 76 seconds.

\section{Conclusion And Outlook}\label{sec3}
In this paper, we have shown a new effective method on calculating syzygies. With a significant feedback in each step, our method can safely protect the \emph{irreducibility} of our basis, which can not be given by other methods in current study. Evidence is given in the paper to show that this algorithm is effective for general ideals, and hints are obtained on the generalizations to modules (although currently we cannot guarantee the full rigorous irreducibility of syzygies for modules). Through this effective method, the idea that mathematical structures of scattering amplitudes are featured with different syzygies can come to practical use via some simple steps of programming.

As a result, physics can be read off through this useful method. As an application, this paper mainly describes s specific context, namely, the IBP relations of a specific two-loop diagram in a Yang-Mills field theory. However, via the general illustrations on the method given before, one can easily simplify all possible IBP relations for a general diagram. Thus, we gain the general and systematic approach on how to simplify the IBP relations and determine irreducible integrals.

This method is fundamental enough that can be generalized to a wide range of applications in the study of scattering amplitudes. First, one can use this method to investigate some other theories and investigate the irreducible integrals of diagrams. Second, tree-level amplitudes or the integrands of loop-level amplitudes can be taken as generators for an ideal, so all possible irreducible relations beyond degree-zero (KK relations) and degree-one (BCJ relations) may be found. One can also use this method to explore simplifications of Grassmannian integral form, and construction of loop amplitudes from unitarity cuts. Furthermore, this algorithm and its ideas, can be used in all possible areas in mathematics and physics that need to simplify complicated algebraic relations (namely, syzygies) of several polynomial (rational) functions. We leave them to future works.

\section*{Acknowledgments}
GC, RX and HZ thank  K. Larsen and Y. Zhang for useful comments and kind suggestions. Useful discussions with Y. E. Cheung, Y. Wang and S. Zhou are gratefully acknowledged. GC, RX and HZ have been supported by the Fundamental Research Funds for the Central Universities under contract 020414340080, NSF of China Grant under contract 11405084, the Open Project Program of State Key Laboratory of Theoretical Physics, Institute of Theoretical Physics, Chinese Academy of Sciences, China (No.Y5KF171CJ1) and the Jiangsu Ministry of Science and Technology under contract BK20131264. JL and YZ will thank the department of physics in Nanjing University for hosting. We also thank Y. Gao, T. Han for hospitality and Key Laboratory of Theoretical Physics for hosting. We thank Tianheng Wang for lots of help in revising our manuscript.

\end{document}